\documentclass{article}
\usepackage{ijcai26}

\usepackage{times}
\usepackage{url}
\usepackage[hidelinks]{hyperref}
\usepackage[utf8]{inputenc}
\usepackage[small]{caption}
\usepackage{graphicx}
\usepackage{amsmath}
\usepackage{amsthm}
\usepackage{booktabs}
\usepackage{algorithm}
\usepackage{algorithmic}

\usepackage{amsfonts,amssymb,mathtools}
\usepackage{dsfont}
\numberwithin{equation}{section}

\usepackage[thinc]{esdiff}
\usepackage{xspace}
\usepackage{packages/macros-common}
\usepackage{packages/macros-project-specific}
\usepackage{packages/mytheorems}

\title{Optimally Curating an Event}

\author{
Mohammad Hajiaghayi$^{1,2}$
\and
S\'{e}bastien Lahaie$^2$\and
Mohammad Mahdavi$^1$\And
Suho Shin$^1$\\
\affiliations
$^1$University of Maryland\\
$^2$Google Research\\
\emails
hajiagha@umd.edu, slahaie@google.com, mahdavi@umd.edu, suhoshin@umd.edu
}

\begin{document}

\maketitle

\begin{abstract}
    We study the optimal design of a self-financing event, a problem that requires balancing the recruitment of costly, positive-value participants with revenue-generating agents who may impose negative values on the event.
We introduce a novel two-sided mechanism design framework with plus agents equipped with private costs and positive impact, and minus agents with private values and negative impact upon inclusion in the event, to maximize the overall quality of the event under a budget-balanced (BB) constraint so that the designer does not run a deficit.

We conduct a comprehensive study on the theory of optimal event curation on various utility functions.
For additive utility, we fully characterize the optimal incentive-compatible Bayesian mechanism under both the ex-ante and ex-post BB constraints.
For submodular utility, we propose an ex-ante BB mechanism that achieves a constant-factor approximation to the optimal Bayesian mechanism.
\end{abstract}

\section{Introduction}
Mechanism design, especially auction design, has been one of the most foundational frameworks in recent decades, with numerous applications to online advertisements, spectrum auctions, and crowdsourcing platforms.
Its importance was highlighted by the 2007 Nobel Prize in Economics awarded to Myerson, Hurwicz, and Maskin for their foundational work on mechanism design.
In particular, two-sided mechanisms, also known as double auctions, serve as a primary tool to transfer goods between agents to maximize social efficiency.
Most of the literature, however, has largely centered on two canonical objectives: maximizing revenue for the broker who governs the trade~\cite{myerson1983efficient,kuang2023profit}, or maximizing social welfare~\cite{deng2022approximately}, typically defined as a function of agents' private valuations and payments.
On the other hand, many modern platforms may realize arbitrary utilities depending on who is involved in the trades, independent of the agents' private valuations or transfers between them.
Nevertheless, there have been only a few works studying such general objective functions, mostly in procurement auction frameworks such as budget-feasible mechanism design~\cite{singer2010budget,anari2014mechanism,balkanski2016bayesian} or all-pay auctions for crowdsourcing contests~\cite{chawla2019optimal,dipalantino2009crowdsourcing}.

We discuss two representative examples to motivate our problem.
Further applications are discussed in Appendix~\ref{apd:fur-appl}.

\paragraph{Example 1: Conference curation}
As an illustrative example, consider the challenge faced by the organizer of a prestigious event, such as an academic conference.
The success of the event depends on how many attendees it attracts, or alternatively, how much welfare it provides to attendees.
On one hand, the organizer seeks to attract renowned keynote speakers, popular artists, or flagship exhibitors.
These participants generate significant positive value and draw more attendees, but they often require invitation fees, travel stipends, or other costs to secure their involvement.
On the other hand, the organizer must fund these costs by securing corporate sponsorships or inviting advertisers.
These participants provide essential revenue but may detract from the overall experience if their presence becomes overwhelming or overly commercial.
The organizer's ultimate goal is to curate the optimal mix of these agents to maximize the attendees' utilities, ensuring that the organizer does not run a deficit by implementing the event.

\paragraph{Example 2: Double auction for LLMs}
As another example, consider an LLM provider, such as Google with its Gemini model, which aims to deliver the most helpful and accurate responses to users. On one hand, to enhance the quality of its service, the provider can procure access to high-quality, up-to-date information from various \emph{plus agents} such as news publishers, academic journals, or specialized data providers. Integrating these documents into a retrieval-augmented generation (RAG) framework allows the LLM to reduce hallucinations and provide more relevant answers, creating significant positive utility for the user. However, these content providers are strategic agents with private costs, and the LLM provider must design a mechanism to pay them for the right to use their data. On the other hand, the provider can fund these content acquisition costs by selling space within the LLM's output to \emph{minus agents}, \ie advertisers.\footnote{There has been increasing attention to integrating ads into LLM outputs, \eg see~\cite{duetting2024mechanism,dubey2024auctions,hajiaghayi2024ad,balseiro2025position}.} These advertisers have private valuations for reaching the user and provide the necessary revenue, but their inclusion may detract from the user experience and the perceived objectivity of the answer. The LLM provider's challenge is to design a joint mechanism, balancing a procurement system for documents with an ad auction, to curate the optimal blend of retrieved knowledge and advertisements, thereby maximizing user satisfaction in a budget-balanced, self-financing ecosystem.

\paragraph{Our model}
To formally capture this complex trade-off, we introduce a novel \emph{two-sided mechanism design} problem for the optimal curation of agents with exogenous utilities, providing the first general framework for two-sided mechanisms where the designer's objective is neither revenue nor social welfare, but originates from an exogenous (possibly combinatorial) utility function.

Our model consists of two types of agents: \emph{plus agents}, who contribute positively to the event's value but incur a private cost to participate, and \emph{minus agents}, who have a private valuation for joining and are willing to pay, but may impose a negative externality on the event.
The mechanism designer aims to maximize an exogenous utility function that maps an allocation to a nonnegative real value, which is monotone increasing in plus agents and monotone decreasing in minus agents, while satisfying (i) dominant-strategy incentive compatibility (DSIC) or Bayesian incentive compatibility (BIC), \ie agents report their true values; (ii) individual rationality (IR), \ie agents are not worse off by participating in the trade; and (iii) budget balance, \ie the mechanism does not run a deficit.
We consider an additive utility function that can be represented as a simple sum of individual contributions, as well as a more complex combinatorial function capturing synergistic or substitution effects between agents.
We analyze this curation problem primarily in the Bayesian setting, where the designer has statistical knowledge of agent types.

To highlight our main results:
\begin{enumerate}
\item For the case of additive utilities in the Bayesian setting, we provide a complete characterization of the optimal DSIC and IR mechanism under ex-ante BB constraint, and the optimal BIC and IR mechanism under ex-post BB constraint, both of which are efficiently computable.
\item For separable submodular utility, we provide a posted pricing mechanism under regular distributions with polynomial running time that is a constant-factor approximation to the optimal Bayesian mechanism.
\end{enumerate}

In what follows, we overview our main results and techniques.
Then, we introduce the formal problem setup, results, and their implications.
Additional results, including irregular-distribution analysis (Appendix~\ref{sec:add-expost-irr}), an additive-utility extension under matroid constraints (Appendix~\ref{sec:additive-matroid}), computational details for the submodular setting (Appendix~\ref{apd:compute}), a simpler submodular mechanism (Appendix~\ref{apd:simple}), comparative statics (Appendix~\ref{apd:comp-statics}), and prior-free analysis for additive utility under a large-market assumption (Appendix~\ref{apd:prior-free}), are provided in the appendix.

\subsection{Our Results and Techniques}
We briefly introduce our main results and the techniques used to obtain them.
We begin with the additive utility setting in Section~\ref{sec:additive}, where each plus agent imposes a fixed positive surplus $p_i \ge 0$ upon selection, and each minus agent induces a fixed negative impact $n_j \le 0$ upon inclusion.\footnote{We refer to Section~\ref{sec:model} for the detailed problem setup.}
In this case, we provide a complete characterization of the optimal mechanism that maximizes the expected utility while remaining ex-ante BB, DSIC, and IR.
We extend the classic Myersonian framework by characterizing the payment rule and deriving the two-sided monotonicity of the allocation rule for any DSIC mechanism.
This allows us to rewrite the budget-balancedness constraint as a function of the virtual costs and valuations of plus and minus agents, respectively.
Then, we write the Lagrangian of the designer's resulting optimization problem.
Finally, using the Karush-Kuhn-Tucker (KKT) conditions of the Lagrangian, we show that we can find the closed-form solution for the optimal Bayesian mechanism by determining the optimal Lagrange multiplier of the budget constraint, known as the shadow price, and deriving the resulting allocation function.
Formal proofs for the additive results appear in Appendix~\ref{apd:additive-proofs}; Appendix~\ref{sec:additive-matroid} discusses an extension under matroid constraints.
Then, by applying the interim transformation technique by~\cite{brustle2017approximating}, we show that our mechanism extends to the optimal BIC mechanism under ex-post BB constraint.
Further discussions on the case of irregular distributions and ex-post BB DSIC mechanisms are provided in Appendix~\ref{sec:add-expost-irr}.

Second, we move beyond the additive utility to the submodular utility case and consider a utility function that is separable into the plus and minus agent sides. Namely, \emph{including} a plus agent and \emph{excluding} a minus agent both increase the utility in a concave manner.
\footnote{In LLM auctions, one ad causes little distraction, but many ads collectively harm quality, motivating our supermodular negative-impact assumption. Since only a few ads are shown to preserve user experience, marginal degradation becomes significant.}
We first argue that the decision version of the offline problem setup (assuming truthful bids) is already NP-hard to approximate within some constant factor, via a reduction from the submodular maximization problem with knapsack constraints.
Accordingly, we provide a simple posted pricing mechanism under regular distributions that can be efficiently computed and achieves a constant-factor approximation to the optimal mechanism.
Overall, our mechanism first generates a set of candidate budgets that can be used for plus agents.
Then, for each candidate budget, we solve the plus and minus agent sides separately to exploit and procure the targeted budget via the multilinear extension of the original optimization problems.
Such multilinear extensions are guaranteed to approximate the optimal performance within a constant factor, thanks to the bounded correlation gap of submodular set functions shown by~\cite{agrawal2010correlation}.
Finally, we select the optimal budget that maximizes the expected utility among the candidates.
Equipped with the fact that the expected utility under the multilinear extension is concave in the budget, we show that the error from using discretized budget candidates can be made small, implying that the resulting maximal mechanism is a good approximation to the optimal one.
Proofs for the submodular results are given in Appendix~\ref{apd:beyond-proofs}; Appendix~\ref{apd:simple} gives a simpler mechanism with weaker guarantees, and Appendix~\ref{apd:compute} discusses computational details and related extensions.
Appendix~\ref{apd:prior-free} studies prior-free mechanisms under a large-market assumption.
\section{Related Work}\label{sec:related}
\paragraph{Budget-feasible mechanism design}
Our problem, in particular the plus agents' side market, has a close connection to the seminal budget-feasible mechanism design introduced by~\cite{singer2010budget}, where a buyer needs to procure a set of items from sellers given an exogenous budget constraint and a set function.
\cite{singer2010budget} show that the standard VCG framework does not work, and provide several constant-factor approximate prior-free mechanisms for submodular utility.
\cite{balkanski2016bayesian} consider a Bayesian mechanism, where the benchmark is the optimal Bayesian mechanism rather than the offline benchmark.
\cite{anari2014mechanism} consider a budget-feasible mechanism with applications to the crowdsourcing market, justifying a \emph{large market assumption} to derive stronger approximation guarantees.
Since then, there have been numerous works~\cite{balkanski2022deterministic,jalaly2021simple,han2023triple,klumper2022budget} improving the approximation factors and generalizing the objective function.

\paragraph{Two-sided market}
Plus and minus agents in our problem can be viewed as buyers and sellers in a two-sided market~\cite{mcafee1992dominant,myerson1983efficient}.
While seminal works in the literature~\cite{myerson1983efficient} focus on deriving optimal mechanisms for restricted cases or prove impossibility results, a recent work~\cite{deng2022approximately} first proved that the gains-from-trade, a measure of social efficiency, can be approximated within a constant factor of the first-best outcome.
Since then, there have been many works analyzing approximation factors or weakening the assumptions~\cite{blumrosen2014reallocation,blumrosen2021almost,deng2025approximately}.
Notably, similar to our problem where the mechanism designer aims to maximize a self-interested utility function, there has been a line of works studying the mechanism of the broker---who governs the trade---to maximize the broker's own profit~\cite{kuang2023profit,eilat2021bilateral,hajiaghayi2025gains}.
In stark contrast, in our setting, the designer has an exogenous utility function that can depend (possibly combinatorially) on the allocation function.

\section{Model}\label{sec:model}
There are $n$ \emph{plus}-agents (\eg keynote speakers) and $m$ \emph{minus}-agents (\eg advertisers), denoted by $\cN = \{1,\ldots, n\}$ and $\cM = \{n+1,\ldots, n+m\}$, respectively.
We consider a Bayesian setting where each plus agent $i$ has a private cost $c_i$ sampled from a distribution $G_i$ supported on $[c_{\min},c_{\max}]$ which represents how much it wants to be paid to participate in the event.
Analogously, each minus agent $j$ has a private valuation $v_j$ sampled from a distribution $F_j$ supported on $[v_{\min}, v_{\max}]$, denoting the willingness-to-pay to join the event.
Let $F$ and $G$ be the product distributions of $F_j$'s and $G_i$'s.
We will use boldface to denote vectors, \eg $\bc$ and $\bv$ for vectors of costs and values, respectively.

We consider a direct mechanism $(\bx, \bt)$ that maps the reported bids $(\tilde{\bc}, \tilde{\bv})$ from the agents to an allocation vector $\bx  = (x_1,\ldots, x_{n+m})\in \{0,1\}^{n+m}$ denoting whether each agent is selected to join the event and a payment vector $\bt = (t_1,\ldots, t_{n+m}) \in \R_{\ge 0}^{n+m}$ representing how much money each agent should pay or will be paid.\footnote{One can easily check that it is without loss of generality to consider direct mechanisms by the revelation principle.}
In particular, we will assume that $t_i \ge 0$ will be paid to plus agent $i$ and $t_j\ge 0 $ will be charged to minus agent $j$.

Given a direct mechanism $(\bx,\bt)$, we write $\bb = (b_1,\ldots, b_{n+m})$ to denote the reported bids of the agents and $b_{-i}$ to represent the vector of reports from all agents except agent $i$.
Accordingly, we write $t_i(\bb) = t_i(b_i, b_{-i})$ and $x_i(\bb) = x_i(b_i, b_{-i})$ to denote the payment and allocation specific to agent $i$ reporting $b_i$ given others' bids $b_{-i}$.

\paragraph{Utilities}
Given an allocation vector $\bx \in \{0,1\}^{n+m}$, the utility of the mechanism designer is defined by a set function $u:\{0,1\}^{n+m} \to \R_{\ge 0}$.
In particular, $u(\cdot)$ is \emph{monotone increasing} over plus agents in a sense that for two allocation vectors $\bx$ and $\by$ such that $\bx_i \le \by_i$ for $i \in \cN$ and $\bx_j = \by_j$ for $j \in \cM$, we have $u(\bx) \le u(\by)$, \ie allocating more plus agents does not decrease the utility.
Similarly, $u(\cdot)$ is \emph{monotone decreasing} over minus agents in a sense that for two allocation vectors $\bx$ and $\by$ such that $\bx_j \le \by_j$ for $j \in \cM$ and $\bx_i = \by_i$ for $i \in \cN$, we have $u(\bx) \ge u(\by)$, \ie allocating more minus agents does not increase the utility.
In words, plus agents increase the utility but it requires cost to hire them, whereas minus agents decrease the utility but one can extract revenue from them.
We write $\one_n$ to denote $(1,1,\ldots, 1) \in \R^n$, and similarly $\zero_n$ to denote $(0,0,\ldots, 0) \in \R^n$.
We assume that the utility function is always nonnegative, or equivalently $u(\zero_n, \one_m) \ge 0$.
We often abuse notation and write a set as the input, \eg $u(\cN)$ or $u(\cM)$, to denote the utility value given corresponding allocation vector.

A set function $f: 2^{S} \mapsto \R_{\ge 0}$ with finite set $S$ is submodular if $f(X \cup \{x\}) - f(X) \ge f(Y \cup \{x\}) - f(Y)$ for any $X\subseteq Y \subseteq S$ and $x \in S \setminus Y$, and supermodular if the inverse of the inequality holds.
We say $f$ is monotone if $f(X) \le f(Y)$ whenever $X \subseteq Y$.
Note that these notations naturally extend to our utility function $u(\cdot)$.

\paragraph{Desiderata}
Our objective is to design a mechanism that maximizes the utility while respecting several desirable properties.
First, given others' bids $b_{-i}$, plus agent $i$'s payoff from reporting $\tilde{c}_i$ when its true cost is $c_i$ is:
\begin{align*}
    U_i(\tilde{c}_i | c_i, b_{-i}) = t_i(\tilde{c}_i, b_{-i}) - c_i \cdot x_i(\tilde{c}_i, b_{-i})
\end{align*}
Similarly, for fixed $b_{-j}$, minus agent $j$'s utility from reporting $\tilde{v}_j$ when its true value is $v_j$ is:
\begin{align*}
    U_j(\tilde{v}_j | v_j, b_{-j}) = v_j \cdot x_j(\tilde{v}_j, b_{-j}) - t_j(\tilde{v}_j, b_{-j})
\end{align*}

We say a mechanism is \emph{dominant strategy incentive-compatible} (DSIC) if deviating from truthful reports does not increase the agents' payoffs.
Formally, for any plus agent $i$, any true cost $c_i$, any lie $c'_i$, and any reports $b_{-i}$:
\begin{align} \label{eq:dsic_positive}
    t_i(c_i, b_{-i}) - c_i x_i(c_i, b_{-i}) \ge t_i(c'_i, b_{-i}) - c_i x_i(c'_i, b_{-i}),
\end{align}
and for any minus agent $j$, any true value $v_j$, any lie $v'_j$, and any reports $b_{-j}$:
\begin{align} \label{eq:dsic_negative}
    v_j x_j(v_j, b_{-j}) - t_j(v_j, b_{-j}) \ge  v_j  x_j(v'_j, b_{-j}) - t_j(v'_j, b_{-j}).
\end{align}

Further, we aim to design an \emph{individually rational} mechanism such that no agent gets a negative payoff from participating in the mechanism.
Formally, for plus agent $i$ with cost $c_i$ and minus agent $j$ with value $v_j$, the following must hold, respectively:
\begin{align*}
    t_i(c_i, b_{-i}) - c_i \cdot x_i(c_i, b_{-i}) &\ge 0
    \\
    v_j \cdot x_j(v_j, b_{-j}) - t_j(v_j, b_{-j}) &\ge 0
\end{align*}

Our mechanism designer should not run a deficit by running the mechanism, \ie it should be \emph{ex-ante (weakly) budget-balanced} (BB):
\begin{align*}
    \int \sum_{j \in \cM}  t_j(\bv, \bc) \,dF(\bv)dG(\bc) \ge \int \sum_{i \in \cN} t_i(\bv, \bc)\, dF(\bv)dG(\bc).
\end{align*}
We often consider \emph{ex-post budget-balancedness} (ex-post BB) where the budget constraint holds for every realization of valuations:
\begin{align*}
    \sum_{j \in \cM}  t_j \ge \sum_{i \in \cN} t_i.
\end{align*}

We finally introduce virtual valuations and costs.
Let $f_j$ and $g_i$ be the density functions of $F_j$ and $G_i$, respectively.
\begin{definition}[Virtual valuation and cost]
    The virtual cost of plus agent $i$ is defined as:
    \begin{align*}
        \psi_i(c_i) = c_i + \frac{G_i(c_i)}{g_i(c_i)}.
    \end{align*}
    The virtual valuation of minus agent $j$ is defined as:
    \begin{align} \label{eq:virtual_valuation}
        \phi_j(v_j) = v_j - \frac{1 - F_j(v_j)}{f_j(v_j)}
    \end{align}
\end{definition}
We say the agents' distributions are \emph{regular} if $\psi_i(c_i)$ and $\phi_j(v_j)$ are monotone increasing for every $i$ and $j$.
Unless specified otherwise, we will restrict our focus to the regular distributions.

\section{Optimal Mechanism for Additive Utility}\label{sec:additive}
We first consider a simple setting where the utility function is additive over each agent.
Formally, the utility function is additive if it can be written as
\begin{align*}
    u(\bx) =  \sum_{i \in \cN} p_i x_i + \sum_{j \in \cM} x_j n_j,
\end{align*}
where $p_i \ge 0$ and $n_j \le 0$.\footnote{This is for the sake of exposition, while our results easily generalize without such restrictions.}
That is, each plus agent $i$ induces a fixed positive impact $p_i$ upon its selection whereas each minus agent $j$ incurs a negative effect of $n_j$ upon inclusion in the event.
Note that it is without loss of generality to consider the domain of the allocation vector to be $[0,1]^{n+m}$, due to the linearity of expectation given that the utility function is additive.

\subsection{Optimal Ex-ante BB DSIC Mechanism}
The following generalizes Myerson's lemma:
\begin{lemma}\label{lm:myerson-monotone}
    For any DSIC mechanism, the allocation rule $\bx$ is monotone decreasing in plus agent $i$'s cost, and monotone increasing in minus agent $j$'s value.
    Further, given others' bids $b_{-i}$ and $b_{-j}$, the payment rule for plus agent $i$ and minus agent $j$ can uniquely be characterized by:
    \begin{align}
         t_i(c_i, b_{-i}) &= c_i \cdot x_i(c_i, b_{-i}) + \int_{c_i}^{c_{\max}} x_i(s, b_{-i}) \,ds + U_{i,0}, \label{eq:payment_positive}
         \\
        t_j(v_j, b_{-j}) &= v_j \cdot x_j(v_j, b_{-j}) - \int_{v_{\min}}^{v_j} x_j(s, b_{-j}) \,ds - U_{j,0},\label{eq:payment_negative}
    \end{align}
    where $U_{i,0} = U_i(c_{\max}, b_{-i})$ and $U_{j,0} = U_j(v_{\min}, b_{-j})$.
\end{lemma}
The next crucial step is to analyze the \emph{expected} payment from each agent from the designer's perspective.
This allows us to replace payments with virtual costs and valuations.
Under the standard assumption that the integration constants $U_{i,0}$ and $U_{j,0}$ are zero, we have the following.

\begin{lemma}\label{lm:virtual}
    The expected expenditure to plus agent $i$ and expected revenue from minus agent $j$ can be written as:
    \begin{align*}
        \Exu{c_i}{t_i(c_i)} &= \Exu{c_i}{x_i(c_i)\psi_i(c_i)}
        \\
        \Exu{v_j}{t_j(v_j)} &= \Exu{v_j}{x_j(v_j)\phi_j(v_j)}
    \end{align*}
\end{lemma}

We can now combine the characterization of implementable mechanisms with the concept of virtual valuations and costs to solve the designer's problem. The following theorem describes the optimal mechanism.

\begin{theorem}
\label{thm:optimal_mechanism}
Under regular distributions, the mechanism described in Algorithm~\ref{alg:optimal_additive} is DSIC, IR, ex-ante BB, and maximizes the utility.
\end{theorem}
The proof proceeds by first expressing the budget-balancedness constraint in terms of virtual valuations and costs.
Then, we write the Lagrangian of the designer's optimization problem, which is solely a function of allocation probabilities, $p_i$'s and $n_j$'s, and virtual valuations and costs.
Then, we can regroup the terms by agent type, along with the corresponding allocation probability.
The KKT stationarity condition gives us the allocation rule described in the algorithm.
Further, the complementary slackness guarantees that the budget constraint must bind or the Lagrange multiplier should be zero, resulting in our choice of shadow price $\lambda^*$.

\begin{algorithm}[tb]
\caption{Optimal Mechanism for Additive Utility}
\label{alg:optimal_additive}
\textbf{Input}: Reported bids $\bc$ and $\bv$
\begin{algorithmic}[1]
\STATE Find the smallest shadow price $\lambda^* \ge 0$ such that the ex-ante budget balance constraint is met:
$$ \Exu{}{\sum_{j \in \cM} t_j(\lambda^*)} \ge \Exu{}{\sum_{i \in \cN} t_i(\lambda^*)} $$

\STATE Define the final allocation rule $\bx^* = \bx(\lambda^*)$ based on the found shadow price:
\STATE \quad $x_i^* \leftarrow \mathbb{I}[p_i - \lambda^* \psi_i(c_i) \ge 0]$ for all $i \in \cN$.
\STATE \quad $x_j^* \leftarrow \mathbb{I}[n_j + \lambda^* \phi_j(v_j) \ge 0]$ for all $j \in \cM$.

\STATE Compute the final payments $\bt^* = \bt(\lambda^*)$ based on the allocation rule $\bx^*$ using~\eqref{eq:payment_positive} and~\eqref{eq:payment_negative}.

\STATE \textbf{return} $(\bx^*, \bt^*)$
\end{algorithmic}
\end{algorithm}

\subsection{Optimal Ex-post BB BIC Mechanism}
\label{sec:bic-expost}

In this section we formalize the (now standard) payment transformation that converts any DSIC \& ex-ante BB mechanism into a BIC \& \emph{ex-post} BB (indeed, SBB) mechanism \emph{without changing the allocation rule}. We then use it to argue that the DSIC \& ex-ante BB optimal mechanism from Theorem~\ref{thm:optimal_mechanism} yields an \emph{optimal} BIC \& ex-post BB mechanism after the transformation.

Recall that ex-ante BB requires
\begin{align*}
\Ex{ \sum_{j\in\cM} t_j(\bv,\bc) } \;\ge\; \Ex{ \sum_{i\in\cN} t_i(\bv,\bc) }.
\end{align*}
Ex-post BB requires the above to hold pointwise for every realized profile; \emph{strong} budget balance (SBB) requires equality ex post:
\begin{align*}
\sum_{j\in\cM} t_j(\bv,\bc) \;=\; \sum_{i\in\cN} t_i(\bv,\bc)\qquad\text{for all }(\bv,\bc).
\end{align*}

\paragraph{Interim Transformation}
We state the transformation theorem first. It preserves the allocation rule and any objective that depends only on \(\bx\).

\begin{lemma}[Payment transformation by~\cite{brustle2017approximating}]
\label{lem:cai-transform}
Let \((\bx,\bt)\) be a mechanism that is BIC, interim IR, and ex-ante BB with a nonnegative payment rule (i.e., \(t_i(\cdot)\ge 0\) for \(i\in\cN\) and \(t_j(\cdot)\ge 0\) for \(j\in\cM\)). Then there exists another payment rule \(\bt'\) such that:
\begin{enumerate}
    \item \emph{Allocation preserved:} \((\bx,\bt')\) uses the same allocation rule \(\bx\).
    \item \emph{BIC \& interim IR preserved:} \((\bx,\bt')\) remains BIC and interim IR, and each agent's interim expected payment is unchanged.
    \item \emph{Ex-post SBB:} For every profile \((\bv,\bc)\),
    \begin{align*}
    \sum_{j\in\cM} t'_j(\bv,\bc) \;=\; \sum_{i\in\cN} t'_i(\bv,\bc).
    \end{align*}
\end{enumerate}
Moreover, the converse holds: any BIC, interim IR, SBB mechanism with nonnegative payments can be converted to a BIC, interim IR, ex-ante BB mechanism with the same allocation by an inverse linear transformation of payments.
\end{lemma}

\begin{remark}[On ``same allocation'']
Lemma~\ref{lem:cai-transform} preserves the \emph{ex-post} allocation rule \(\bx(\cdot)\) profile-by-profile; in particular, interim allocation probabilities remain identical. Since all of our objectives in this paper depend only on the allocation (transfers do not enter the objective), the designer's value is unchanged by the transformation.
\end{remark}

\paragraph{Optimality}
We now show that the DSIC \& ex-ante BB optimum characterized in Section~\ref{sec:additive} yields an \emph{optimal} BIC \& ex-post BB mechanism after the transformation in Lemma~\ref{lem:cai-transform}.

\begin{theorem}[Optimal BIC \& ex-post BB]
\label{thm:bic-expost-optimal}
In the additive setting under regular distributions, let \((\bx^*,\bt^*)\) be the DSIC, IR, ex-ante BB mechanism from Theorem~\ref{thm:optimal_mechanism}. Then there exists a payment rule \(\widehat{\bt}\) such that \((\bx^*,\widehat{\bt})\) is BIC, IR, and ex-post BB (indeed SBB), achieves the same expected utility as \((\bx^*,\bt^*)\), and is \emph{optimal} among all BIC, IR, ex-post BB mechanisms.
\end{theorem}

\begin{corollary}[Equality of optimal values]
\label{cor:equiv-values}
In the additive setting with regular distributions,
\begin{align*}
OPT_{\textnormal{BIC, ex-ante BB}}
\;=\;
OPT_{\textnormal{BIC, ex-post BB}}
\;=\;
\Ex{u(\bx^*)},
\end{align*}
where \(\bx^*\) is the allocation from Theorem~\ref{thm:optimal_mechanism}.
\end{corollary}

\section{Beyond Additive Utility}\label{sec:beyond}
When the utility function $u(\cdot)$ is not additive but combinatorial, \eg submodular, the previous framework does not work.
In particular, even the offline optimization problem, assuming truthful bids, is well-known to be NP-hard to solve for many classes of functions.
Thus, we aim to derive a mechanism that approximates the optimal mechanism.

Let $S$ be the set of selected agents given an allocation vector $\bx$.
Accordingly, we consider the following class of separable functions
\begin{align*}
    u(S) = v_p(S_p) - v_m(S_m),
\end{align*}
where $S_p$ and $S_m$ are the sets of \emph{selected} plus and minus agents, respectively.
Further, we assume $v_p$ is monotone increasing submodular and $v_m$ is monotone decreasing supermodular.
This is practical since the marginal decrease in the audience experience can grow significantly as we assign more advertisements, even though a few of them might be fine.
From theoretical perspectives, this is reasonable since if we alternatively assume $v_m$ is submodular, then the resulting function is the sum of a submodular $v_p$ and a supermodular $-v_m$, which admits no polynomial approximation algorithm.\footnote{In general, maximizing a monotone supermodular function always requires an exponential number of queries.}

Note that maximizing $u$ is equivalent to maximizing
\begin{align*}
    u(S) + v_m(\cM) = v_p(S_p) + (v_m(\cM) - v_m(S_m)).
\end{align*}
In particular, we can write $u_m(S_m) = v_m(\cM) - v_m(S_m)$ as a monotone increasing submodular function defined on the set of excluded minus agents $S_m$.
Thus, redefining $S_m$ as the set of excluded minus agents for notational convenience, we will consider the utility function $u(S) = u_p(S_p) + u_m(S_m)$ for monotone increasing nonnegative submodular functions $u_p$ and $u_m$ where $S_p$ denotes the set of selected plus agents, and $S_m$ the set of excluded minus agents.

\subsection{Budget-enumeration Posted Pricing}
We first introduce our main mechanism, the \emph{budget-enumeration posted pricing} (\BEPP) mechanism, presented in Algorithm~\ref{alg:bepp}.
It takes a discretization parameter $\eps > 0$ which controls the approximation factor as well as its time complexity.
The following theorem formalizes the performance of the \BEPP mechanism.
\begin{theorem}\label{thm:bepp}
For any $\eps > 0$, the \BEPP mechanism is DSIC, IR, and ex-ante BB, runs in $\mathrm{poly}(n, m, 1/\eps, \log B_{\max})$ time, and achieves a $(1-1/e - O(\eps))$-approximation to the optimal Bayesian mechanism for separable submodular utility and regular distributions.
\end{theorem}

\paragraph{Proof overview}
We now provide the overall proof steps to show that the \BEPP mechanism is a constant-factor approximation of the optimal Bayesian mechanism for separable submodular utilities. The proof proceeds by defining a series of relaxations for the plus-agent and minus-agent subproblems, establishing an upper bound on the optimal welfare, and then showing that the expected welfare of our posted-price mechanism is within a constant factor of this bound.

We first formally define the two key relaxations of a set function, following the approaches used by~\cite{balkanski2016bayesian}. Let $u: 2^\cO \to \R_{\ge 0}$ be a set function on a ground set $\cO$. Let $\bp = (p_k)_{k \in \cO}$ be a vector of marginal probabilities.
We now define the multilinear extension of $u$.

\begin{algorithm}[tb]
\caption{\BEPP Mechanism}
\label{alg:bepp}
\textbf{Input}: Parameter $\eps > 0$
\begin{algorithmic}[1]
\STATE Compute $B_{\max}$, the maximum expected revenue extractable from the minus agents via Myerson's optimal auction.
\STATE Find minimal $M$ such that $(1+\eps)^M > B_{\max}$.
\STATE Define $\cB \leftarrow \{0\} \cup \{(1+\eps)^i : 0 \le i \le M-1\}$.
\FOR{each budget target $B_k \in \cB$}
    \STATE Find $\bq(B_k)$ from multilinear version of~\eqref{eq:opt-multi-plus}.
    \STATE Find $\br(B_k)$ from multilinear version of~\eqref{eq:opt-multi-minus}.
    \STATE $W^{\mathrm{ml}}(B_k) \leftarrow u_p(\bq(B_k)) + u_m(\br(B_k))$.
\ENDFOR
\STATE Find the best budget $B^* \leftarrow \arg\max_{B_k \in \cB} W^{\mathrm{ml}}(B_k)$.
\STATE Set $t^p_i \leftarrow G_i^{-1}(q_i(B^*))$ for each $i \in \cN$ and $t^m_j \leftarrow F_j^{-1}(r_j(B^*))$ for each $j \in \cM$.
\STATE \textbf{return} posted prices $(\bt^p, \bt^m)$: pay $t^p_i$ to plus agent $i$ upon acceptance, and charge $t^m_j$ from minus agent $j$ upon acceptance.
\end{algorithmic}
\end{algorithm}

\begin{definition}[Multilinear Extension]
The \textbf{multilinear extension} of $u$, denoted $U(\bp)$, is the expected value of the function when each element $k \in \cO$ is included in a set $S$ independently with probability $p_k$.
\begin{align*}
    U(\bp) = \Exu{S \sim \bp}{u(S)} = \sum_{S \subseteq \cO} u(S) \prod_{k \in S} p_k \prod_{k \notin S} (1-p_k).
\end{align*}
\end{definition}

The key feature in analyzing approximation factors derived from multilinear extension is its relation to the following concave relaxation of $u$.

\begin{definition}[Concave Closure]
The \textbf{concave closure} of $u$, denoted $U^+(\bp)$, is the maximum expected value of the function over any \emph{correlated} distribution $\mathcal{D}$ on subsets of $\cO$ with marginal probabilities $\bp$.
\begin{align*}
    U^+(\bp) = \max_{\mathcal{D} \text{ with marginals } \bp} \Exu{S \sim \mathcal{D}}{u(S)}
\end{align*}
\end{definition}
Importantly, note here that the probability vectors can be correlated unlike the multilinear extension.
The connection between these two relaxations is the following notion of correlation gap.
\begin{theorem}[\cite{agrawal2010correlation}]\label{thm:corr_gap}
For any non-negative, monotone submodular function $u(\cdot)$, its multilinear extension and concave closure satisfy:
\begin{align*}
    \min_{\bp} \frac{U(\bp)}{U^+(\bp)} \ge 1 - \frac{1}{e}
\end{align*}
\end{theorem}

We structure the proof as a series of lemmas.
First, we write concave closures of our original optimization problems for each side market.
Define $\bq = (q_1,\ldots, q_n)$ as a vector that represents the marginal probability that plus agent $i$ is selected, and similarly $\br = (r_1,\ldots, r_m)$ as a vector that represents the probability that minus agent $j$ is excluded.
In particular, we consider the following program for the plus agent side for each $B_k$:
\begin{align}
    (\text{PLUS})\quad  & \max_{\bq }  && u_p(\bq) \label{eq:opt-multi-plus}
    \\
    & \text{s.t.}  && \sum_{i \in \cN} q_i G_i^{-1}(q_i) \le B_k, \nonumber
\end{align}
and the following for the minus agent side:
\begin{align}
    (\text{MINUS})\quad  & \max_{\br} \quad &&u_m(\br) \label{eq:opt-multi-minus}\\
    &\text{s.t.} \quad && \sum_{j \in \cM} (1-r_j) F_j^{-1}(r_j) \ge B_k. \nonumber
\end{align}
Multilinear versions of the above programs require the probabilities to be independent, whereas concave closure allows them to be correlated.

We will compare each optimization problem with the expected utilities obtainable from each side market.
Formally, we will consider a class of mechanisms $\Mec(B)$ parameterized by $B$ such that its total expected payment to plus agents is at most $B$, and its total expected revenue from minus agents is at least $B$.
Given a budget $B$, we use the following notation:
\begin{itemize}
    \item $W_p(B)$ and $W_m(B)$: the optimal expected utilities achievable by any DSIC and IR mechanism in $\Mec(B)$, on the plus and minus side respectively. We write $W(B) = W_p(B) + W_m(B)$.
    \item $W_p^+(B)$ and $W_m^+(B)$: the optimal values of (PLUS) and (MINUS) under the concave closure.
    \item $W_p^{\mathrm{ml}}(B)$ and $W_m^{\mathrm{ml}}(B)$: the optimal values of (PLUS) and (MINUS) under the multilinear extension, with $\bq(B)$ and $\br(B)$ denoting the corresponding optimal marginal probabilities.
\end{itemize}
Throughout, we assume without loss of generality (by rescaling costs and values) that $W(x) \le (1+\eps) W(0)$ for all $x \in (0, 1]$.

The performance of any mechanism given a budget, including the optimal one, will be upper-bounded by the concave closure as shown by the following lemma:
\begin{lemma}[Concave relaxation] \label{lem:upper_bound}
Given $B > 0$ and a mechanism $\Mec(B)$ that is DSIC, IR and ex-ante BB, its expected utility $u_p$ from plus agents is upper bounded by $W_p^+(B)$ and $u_m$ (that from minus agents) is upper bounded by $W_m^+(B)$.
\end{lemma}

Then, we consider multilinear versions of the programs (PLUS) and (MINUS), and solving them gives us prices to offer as depicted in Algorithm~\ref{alg:bepp}.
Thanks to Theorem~\ref{thm:corr_gap}, these prices guarantee a constant-factor approximation to the concave closure given $B$.

\begin{lemma}[Posted pricing given a budget] \label{lem:pp_apx}
For any budget $B$, offering prices $\bt^p = G^{-1}(\bq(B))$ to plus agents guarantees an expected utility of at least $(1-1/e) \cdot W_p^+(B)$ from plus agents, and offering $\bt^m = F^{-1}(\br(B))$ to minus agents guarantees an expected utility of at least $(1-1/e) \cdot W_m^+(B)$ from minus agents.
\end{lemma}

Therefore, one can conclude that if we know the expected budget used by the optimal mechanism, the corresponding posted pricing mechanism gives us the desired result.
On the other hand, this is not known in advance, so we generate a finite list of candidate optimal budgets and select the candidate budget that maximizes the expected utility.
Since this list cannot always include the optimal budget due to the continuity of the entire budget space, we require the following concavity of the optimal expected utility with respect to the budget.
\begin{lemma}[Concavity of the optimal utility] \label{lem:concavity}
Under regular agent distributions, the plus-agent's optimal utility $W_p(B)$ and the minus-agent's optimal utility $W_m(B)$ are both concave functions of $B$.
\end{lemma}

The above lemma results in the following lemma, implying that the error from the finiteness of the budget candidates can be made sufficiently small.
\begin{lemma}[Error from budget discretization] \label{lem:budget_error}
Let $W(B) = W_p(B) + W_m(B)$ be the total utility function given $B$. Let $B_k$ and $B_{k+1}$ be consecutive grid points in $\mathcal{B}$ such that $B_k \le B^* \le B_{k+1}$. Then, we have $W(B_k) \ge (1-\eps)W(B^*)$.
\end{lemma}
Combining the lemmas, let $B_{\textnormal{OPT}}$ be the expected budget used by the optimal Bayesian mechanism, and let $B_k, B_{k+1}$ be consecutive grid points bracketing it ($B_k \le B_{\textnormal{OPT}} \le B_{k+1}$). Then
\begin{align*}
W^{\mathrm{ml}}(B^*) &\ge W^{\mathrm{ml}}(B_k) \\&\ge (1-1/e) \cdot W^+(B_k) &&\text{(Theorem~\ref{thm:corr_gap})}\\
&\ge (1-1/e) \cdot W(B_k) &&\text{(Lemma~\ref{lem:upper_bound})}\\
&\ge (1-1/e)(1-\eps) \cdot W(B_{\textnormal{OPT}}) &&\text{(Lemma~\ref{lem:budget_error})}\\
&= (1-1/e)(1-\eps) \cdot \textnormal{OPT},
\end{align*}
which yields the claimed $(1-1/e - O(\eps))$-approximation.

\section{Conclusion}
In this paper, we introduced and analyzed a novel two-sided mechanism design framework for self-financing events, where a designer's goal is to maximize an exogenous utility function by curating a set of participants. This problem requires balancing the recruitment of costly, high-impact plus agents with revenue-generating minus agents who may detract from the event's quality, all under a budget-balanced constraint.

Our analysis addresses this design problem primarily within the Bayesian setting. For the case of additive utilities, we fully characterize the optimal mechanisms that are incentive-compatible, IR, and BB. Extending our results to more general objectives, we then provide constant-factor approximation mechanisms for utility functions that are separable and submodular.

Key avenues for future work include extending our analysis beyond submodular separable utility functions to model the complex interplay between participants. Broadening the scope of prior-free results for this framework, by providing robust guarantees for submodular utilities or by relaxing the large market assumption, could also enhance its practical applicability.

\section*{Acknowledgements}
This work is partially supported by DARPA expMath, NSF AF:Small \#2218678, NSF AF:Small \#2114269, Army-Research Laboratory (ARL) \#W911NF2410052, and MURI on Algorithms, Learning and Game Theory.

\bibliographystyle{named}
\bibliography{ref}

\appendix
\onecolumn

\section{Further applications}\label{apd:fur-appl}
Notably, our framework subsumes broader applications beyond budget-balanced event curation.
We introduce several examples to motivate the practicality of our approach.

\paragraph{Carbon market}
Our framework also models the complex trade-offs in modern environmental policy. Consider a regulator tasked with reducing carbon emissions in a budget-neutral way. The success of their policy is measured by the net reduction in atmospheric carbon. On one hand, the regulator seeks to procure carbon abatement services from a market of plus agents. These agents, ranging from landowners paid to reforest their property to tech companies developing green energy, each face private costs to implement their solutions. While these initiatives provide direct positive environmental utility, they require significant public expenditure, compelling the regulator to design efficient procurement auctions to buy the most abatement for their budget. On the other hand, the regulator needs to fund these subsidies by regulating \emph{minus agents}---the existing industrial polluters. By designing a carbon market, the regulator generates revenue from firms that have a high private value for emitting, but this activity imposes a negative environmental externality. The regulator's ultimate goal is to design this market to find the optimal carbon price, curating a portfolio of emissions allowances and abatement projects that maximizes environmental quality, ensuring the system is self-financing.

\paragraph{Urban planning}
As another application, consider the challenge faced by a city planner aiming to foster a vibrant and livable urban environment. The success of a city is often measured by the quality of life it offers its residents, which depends on a careful balance of economic activity and public amenities. On one hand, a city seeks to enhance its appeal by procuring high-value public works and services from a diverse set of plus agents. These agents, such as construction firms contracted to build new parks or non-profit organizations that run cultural centers, each have private costs to complete their projects. While their contributions generate significant positive value for residents, they must be compensated via public funds, forcing the city to design procurement mechanisms to select the best projects without overspending. On the other hand, the city must fund these investments by fostering a strong economic base through \emph{minus agents} such as industrial facilities or high-density commercial developments that provide essential tax revenue but may impose negative externalities. The city planner's ultimate goal is to use tools like land auctions and procurement mechanisms to find the optimal mix of commercial and public projects, maximizing the city's overall livability while ensuring that public investments are sustainably funded by the tax revenue they help generate.

\section{Remaining Proofs in Section~\ref{sec:additive}}\label{apd:additive-proofs}
\subsection{Proof of Lemma~\ref{lm:myerson-monotone}}
\begin{proof}[Proof of Lemma~\ref{lm:myerson-monotone}]
    The proof closely follows the standard proof of Myerson's lemma, but we include the details for completeness.
    We first prove the monotonicity condition for plus agents.
    The DSIC condition~\eqref{eq:dsic_positive} must hold for any pair of reports, say $c_1$ and $c_2$.
    \begin{align*}
        t_i(c_1, b_{-i}) - c_1 \cdot x_i(c_1, b_{-i}) &\ge t_i(c_2, b_{-i}) - c_1 \cdot x_i(c_2, b_{-i}) \\
        t_i(c_2, b_{-i}) - c_2 \cdot x_i(c_2, b_{-i}) &\ge t_i(c_1, b_{-i}) - c_2 \cdot x_i(c_1, b_{-i})
    \end{align*}
    Rearranging both inequalities to isolate the payment difference, $t_i(c_2, b_{-i}) - t_i(c_1, b_{-i})$, yields:
    \begin{align*}
        c_2 [x_i(c_2, b_{-i}) - x_i(c_1, b_{-i})] &\ge t_i(c_2, b_{-i}) - t_i(c_1, b_{-i}) \\
        c_1 [x_i(c_2, b_{-i}) - x_i(c_1, b_{-i})] &\le t_i(c_2, b_{-i}) - t_i(c_1, b_{-i})
    \end{align*}
    Combining these two:
    \begin{align*}
     c_1 [x_i(c_2, b_{-i}) - x_i(c_1, b_{-i})] \le c_2 [x_i(c_2, b_{-i}) - x_i(c_1, b_{-i})],
    \end{align*}
    which is equivalent to
    \begin{align*}
        (c_2 - c_1) [x_i(c_2, b_{-i}) - x_i(c_1, b_{-i})] \le 0.
    \end{align*}
    Assuming $c_2 > c_1$, then $x_i(c_2, b_{-i}) \le x_i(c_1, b_{-i})$, \ie it is monotone decreasing on the plus agent's cost.

    The argument for minus agents is symmetric. The DSIC condition from Equation~\eqref{eq:dsic_negative} must hold for any two reported valuations, say $v_1$ and $v_2$.
    \begin{align*}
        v_1 \cdot x_j(v_1, b_{-j}) - t_j(v_1, b_{-j}) &\ge v_1 \cdot x_j(v_2, b_{-j}) - t_j(v_2, b_{-j}) \\
        v_2 \cdot x_j(v_2, b_{-j}) - t_j(v_2, b_{-j}) &\ge v_2 \cdot x_j(v_1, b_{-j}) - t_j(v_1, b_{-j})
    \end{align*}
    Rearranging both inequalities to isolate the payment difference $t_j(v_2, b_{-j}) - t_j(v_1, b_{-j})$:
    \begin{align*}
        v_1 [x_j(v_2, b_{-j}) - x_j(v_1, b_{-j})] &\le t_j(v_2, b_{-j}) - t_j(v_1, b_{-j}) \\
        v_2 [x_j(v_2, b_{-j}) - x_j(v_1, b_{-j})] &\ge t_j(v_2, b_{-j}) - t_j(v_1, b_{-j})
    \end{align*}
    Combining these two inequalities yields:
    \begin{align*}
        v_1 [x_j(v_2, b_{-j}) - x_j(v_1, b_{-j})] \le v_2 [x_j(v_2, b_{-j}) - x_j(v_1, b_{-j})]
    \end{align*}
    which is equivalent to
    \begin{align*}
        (v_2 - v_1) [x_j(v_2, b_{-j}) - x_j(v_1, b_{-j})] \ge 0.
    \end{align*}
    Assuming $v_2 > v_1$, this implies that $x_j(v_2, b_{-j}) \ge x_j(v_1, b_{-j})$, i.e., the allocation rule is monotone increasing in the minus agent's valuation.

    \paragraph{Payment rule}
    Recall that
    \begin{align*}
        U_i(\tilde{c}_i | c_i, b_{-i}) = t_i(\tilde{c}_i, b_{-i}) - c_i \cdot x_i(\tilde{c}_i, b_{-i}).
    \end{align*}
    Let $U_i(c_i, b_{-i})$ be the plus agent $i$'s equilibrium utility.
    By the Envelope Theorem:
    \begin{align*}
        \frac{d U_i(c_i, b_{-i})}{d c_i} &= -x_i(c_i, b_{-i})
    \end{align*}
    Integrating from $c_i$ to $c_{\max}$ gives the agent's utility:
    \begin{align*}
        U_i(c_i, b_{-i}) &= \int_{c_i}^{c_{\max}} x_i(s, b_{-i}) \,ds + U_i(c_{\max}, b_{-i})
    \end{align*}
    Recalling that
    \begin{align*}
        U_i(c_i, b_{-i})  = t_i(c_i,b_{-i}) - c_i \cdot x_i(c_i, b_{-i}),
    \end{align*}
    we finally have
    \begin{align*}
        t_i(c_i, b_{-i}) = c_i \cdot x_i(c_i, b_{-i}) + \int_{c_i}^{c_{\max}} x_i(s, b_{-i}) \,ds + U_i(c_{\max}, b_{-i}).
    \end{align*}

    Let $U_j(v_j, b_{-j})$ be the minus agent $j$'s equilibrium utility given their true valuation $v_j$ and others' bids $b_{-j}$. By the Envelope Theorem, the derivative of their utility with respect to their type is:
    \begin{align*}
        \frac{d U_j(v_j, b_{-j})}{d v_j} &= x_j(v_j, b_{-j})
    \end{align*}
    Integrating from $v_{\min}$ to $v_j$ gives the agent's total utility:
    \begin{align*}
        U_j(v_j, b_{-j}) &= \int_{v_{\min}}^{v_j} x_j(s, b_{-j}) \,ds + U_j(v_{\min}, b_{-j})
    \end{align*}
    Recalling the definition of utility for a minus agent,
    \begin{align*}
        U_j(v_j, b_{-j})  = v_j \cdot x_j(v_j, b_{-j}) - t_j(v_j,b_{-j}),
    \end{align*}
    we can solve for the payment rule $t_j$ by substitution:
    \begin{align*}
        t_j(v_j, b_{-j}) = v_j \cdot x_j(v_j, b_{-j}) - \left( \int_{v_{\min}}^{v_j} x_j(s, b_{-j}) \,ds + U_j(v_{\min}, b_{-j}) \right).
    \end{align*}
    This completes the characterization of the payment rule for minus agents.
\end{proof}

\subsection{Proof of Lemma~\ref{lm:virtual}}
\begin{proof}[Proof of Lemma~\ref{lm:virtual}]
    First, if $c_i = c_{\max}$ or $v_j = 0$, assigning positive allocation probability cannot improve the objective in an optimal mechanism. Thus, without loss of generality, we assume $U_{i,0} = U_{j,0} = 0$.

    From Equation~\eqref{eq:payment_negative}, the payment is:
    \begin{align*}
        t_j(v_j) = v_j \cdot x_j(v_j) - \int_{v_{\min}}^{v_j} x_j(s) \,ds
    \end{align*}
    The designer's expected revenue from agent $j$ is:
    \begin{align*}
        \mathbb{E}_{v_j}[t_j(v_j)] &= \mathbb{E}_{v_j} [ v_j x_j(v_j) ] - \mathbb{E}_{v_j} \left[ \int_{v_{\min}}^{v_j} x_j(s) \,ds \right]
    \end{align*}
    Again, we use integration by parts on the second term. Note that this is equivalent to
    \begin{align*}
        \int_{v_{\min}}^{v_{\max}} \left( \int_{v_{\min}}^{v_j} x_j(s) \,ds \right) f_j(v_j) \,dv_j.
    \end{align*}
    Using integrating by parts:
    \begin{align*}
        &= \left[ \left( \int_{v_{\min}}^{v_j} x_j(s, b_{-j}) \,ds \right) \cdot \left( -(1-F_j(v_j)) \right) \right]_{v_{\min}}^{v_{\max}} - \int_{v_{\min}}^{v_{\max}} \left( -(1-F_j(v_j)) \right) x_j(v_j, b_{-j}) \,dv_j
        \\
        &=(0-0) + \int_{v_{\min}}^{v_{\max}} x_j(v_j, b_{-j}) (1 - F_j(v_j)) \,dv_j
        \\
        &= \mathbb{E}_{v_j} \left[ x_j(v_j) \frac{1 - F_j(v_j)}{f_j(v_j)} \right]
    \end{align*}
    Substituting this back:
    \begin{align*}
        \mathbb{E}_{v_j}[t_j(v_j)] &= \mathbb{E}_{v_j} [ v_j x_j(v_j) ] - \mathbb{E}_{v_j} \left[ x_j(v_j) \frac{1 - F_j(v_j)}{f_j(v_j)} \right] \\
        &= \mathbb{E}_{v_j} \left[ x_j(v_j) \left( v_j - \frac{1 - F_j(v_j)}{f_j(v_j)} \right) \right]
    \end{align*}

    The derivation for the plus agents' expected payment is symmetric, showing that their expected cost equals their expected virtual cost.
    From Equation~\eqref{eq:payment_positive}, assuming $U_{i,0}=0$, the payment to agent $i$ is:
    \begin{align*}
        t_i(c_i) = c_i \cdot x_i(c_i) + \int_{c_i}^{c_{\max}} x_i(s) \,ds
    \end{align*}
    The designer's expected cost for agent $i$ is therefore:
    \begin{align*}
        \mathbb{E}_{c_i}[t_i(c_i)] &= \mathbb{E}_{c_i} [ c_i x_i(c_i) ] + \mathbb{E}_{c_i} \left[ \int_{c_i}^{c_{\max}} x_i(s) \,ds \right]
    \end{align*}
    We use integration by parts on the second term. Let $g_i(\cdot)$ and $G_i(\cdot)$ be the PDF and CDF of costs for agent $i$.
    \begin{align*}
        \mathbb{E}_{c_i} \left[ \int_{c_i}^{c_{\max}} x_i(s) \,ds \right] &= \int_{c_{\min}}^{c_{\max}} \left( \int_{c_i}^{c_{\max}} x_i(s) \,ds \right) g_i(c_i) \,dc_i \\
    \end{align*}
    Using integration by parts with $u = \int_{c_i}^{c_{\max}} x_i(s) \,ds$ and $dv = g_i(c_i) \,dc_i$, we get $du = -x_i(c_i) \,dc_i$ and $v = G_i(c_i)$. This yields:
    \begin{align*}
        &= \left[ \left( \int_{c_i}^{c_{\max}} x_i(s) \,ds \right) G_i(c_i) \right]_{c_{\min}}^{c_{\max}} - \int_{c_{\min}}^{c_{\max}} G_i(c_i) (-x_i(c_i)) \,dc_i \\
        &= (0-0) + \int_{c_{\min}}^{c_{\max}} x_i(c_i) G_i(c_i) \,dc_i \\
        &= \mathbb{E}_{c_i} \left[ x_i(c_i) \frac{G_i(c_i)}{g_i(c_i)} \right]
    \end{align*}
    Substituting this result back into the expected cost expression:
    \begin{align*}
        \mathbb{E}_{c_i}[t_i(c_i)] &= \mathbb{E}_{c_i} [ c_i x_i(c_i) ] + \mathbb{E}_{c_i} \left[ x_i(c_i) \frac{G_i(c_i)}{g_i(c_i)} \right] \\
        &= \mathbb{E}_{c_i} \left[ x_i(c_i) \left( c_i + \frac{G_i(c_i)}{g_i(c_i)} \right) \right]
    \end{align*}
    This completes the derivation, showing that the expected payment to a plus agent is their expected virtual cost, $\psi_i(c_i) = c_i + G_i(c_i)/g_i(c_i)$, conditional on being selected.
\end{proof}

\subsection{Proof of Theorem~\ref{thm:optimal_mechanism}}
\begin{proof}[Proof of Theorem~\ref{thm:optimal_mechanism}]
Recall that The designer's objective is to choose an allocation rule $x$ to maximize the expected utility, subject to the constraints that the mechanism is incentive compatible (IC) and budget balanced (BB).
\begin{align*}
\max_{x} \quad &\mathbb{E} \left[ \sum_i x_i(c_i) p_i + \sum_j x_j(v_j) n_j \right] \\
\text{subject to} \quad &(IC,IR) \quad x \text{ is implementable (monotone) and IR} \\
&(BB) \quad \mathbb{E} \left[ \sum_j t_j(v_j) - \sum_i t_i(c_i) \right] \ge 0
\end{align*}

Note that the budget-balancedness constraint can be written as a function over the allocation rule:
\begin{align*}
\mathbb{E} \left[ \sum_j x_j(v_j) \phi_j(v_j) - \sum_i x_i(c_i) \psi_i(c_i) \right] \ge 0
\end{align*}

This is a constrained optimization problem.
We introduce a Lagrange multiplier $\lambda \ge 0$ for the budget-balancedness constraint.
The Lagrangian function $L$ is:
\begin{align*}
    L(\bx, \lambda) &=
    \Ex{ \sum_i x_i p_i + \sum_j x_j n_j} + \lambda \left( \Ex{ \sum_j x_j \phi_j - \sum_i x_i \psi_i }\right)
    \\
    &=
    \Ex{ \sum_i x_i p_i + \sum_j x_j n_j + \lambda \left( \sum_j x_j \phi_j - \sum_i x_i \psi_i \right) }
\end{align*}

For simplicity, let us write $L(\bx, \lambda) = W(\bx) + \lambda B(\bx)$.
We now regroup the terms by agent type:
\begin{align}
    L(\bx, \lambda) = \mathbb{E} \left[ \sum_i x_i(c_i) \underbrace{\left(p_i - \lambda \psi_i(c_i)\right)}_{\text{virtual surplus of plus agent}} + \sum_j x_j(v_j) \underbrace{\left(n_j + \lambda \phi_j(v_j)\right)}_{\text{virtual surplus of minus agent}} \right] \label{eq:lagrangian}
\end{align}

For a solution $(\bx^*, \lambda^*)$ to be optimal, it must satisfy the following KKT conditions:

\begin{enumerate}
    \item (Stationarity) The Lagrangian is maximized at $\bx^*$. For our problem, this means $\bx^*$ must be the allocation rule that maximizes $L(\bx, \lambda^*)$. As can be seen from~\eqref{eq:lagrangian}, this is the rule where each agent is selected if and only if their virtual surplus is non-negative.
    \begin{align*}
        x_i^*(c_i) = 1 &\iff p_i - \lambda^* \psi_i(c_i) \ge 0 \\
        x_j^*(v_j) = 1 &\iff n_j + \lambda^* \phi_j(v_j) \ge 0
    \end{align*}

    \item (Primal Feasibility) The solution $\bx^*$ must be feasible for the original problem. This simply means the budget constraint must be satisfied.
    \begin{align*}
        B(\bx^*) \ge 0
    \end{align*}

    \item (Dual Feasibility) The Lagrange multiplier for a $\ge$ constraint must be non-negative.
    \begin{align*}
        \lambda^* \ge 0
    \end{align*}

    \item (Complementary Slackness) This is the most insightful condition. It states that the multiplier must be zero if the constraint is not binding (i.e., if there is a surplus).
    \begin{align*}
        \lambda^* \cdot B(\bx^*) = 0
    \end{align*}
\end{enumerate}

The KKT conditions, particularly Complementary Slackness, give us a concrete algorithm for finding the optimal $\lambda^*$. Let $B(\lambda)$ denote the expected budget surplus when the allocation rule is determined by a given $\lambda$, \ie $B(\lambda) = B(\bx^*(\lambda))$.

Suppose we set $\lambda = 0$ and the budget constraint binds, \ie $B(\lambda) \ge 0$.
Then this is the optimal allocation rule.
For instance, if $p_i \ge 0$ and $n_j \le 0$, the optimal allocation rule $\bx^*(0)$ is to recruit every plus agent and no minus agent, according to the stationarity condition.

Now suppose we test $\lambda=0$ and find that $B(0) < 0$. This violates the Primal Feasibility condition, so $(\lambda=0)$ cannot be the optimal solution.
From Dual Feasibility, we know $\lambda^* \ge 0$. Since $\lambda^* \neq 0$, we must have $\lambda^* > 0$.
Now we turn to Complementary Slackness: $\lambda^* \cdot B(\bx^*(\lambda^*)) = 0$.
Since we know $\lambda^* > 0$, we must have $B(\bx^*(\lambda^*)) = 0 $.
Note that $B(\lambda)$ is a non-decreasing function of $\lambda$, and $u(\bx(\lambda))$ is a non-increasing function of $\lambda$ since by the stationarity condition, increasing $\lambda$ only recruits less plus agent and more minus agents.
Thus, the optimal mechanism is to find the minimal $\lambda^*$ that binds $B(\lambda^*) =0$, and compute corresponding $\bx^*(\lambda^*)$.

Finally, we must confirm that the derived allocation rule is implementable.
\begin{itemize}
    \item For a plus agent $i$, the allocation rule is $x_i=1 \iff p_i \ge \lambda \psi_i(c_i)$. Since $\psi_i(c_i)$ is assumed to be an increasing function of $c_i$, a higher cost $c_i$ makes this condition harder to satisfy. Thus, $x_i(c_i)$ is non-increasing in $c_i$.
    \item For a minus agent $j$, the rule is $x_j=1 \iff n_j + \lambda \phi_j(v_j) \ge 0$. Since $\phi_j(v_j)$ is an increasing function of $v_j$, a higher valuation $v_j$ makes this condition easier to satisfy. Thus, $x_j(v_j)$ is non-decreasing in $v_j$.
\end{itemize}
Since the allocation rules are monotone, they are implementable. This completes the proof that the mechanism described in Theorem~\ref{thm:optimal_mechanism} is optimal.
\end{proof}

\subsection{Proof of Theorem~\ref{thm:bic-expost-optimal}}

\begin{proof}{Proof of Theorem~\ref{thm:bic-expost-optimal}}
\emph{\((\bx^*,\bt^*)\) is optimal for BIC \& ex-ante BB.}
We write the BIC \& ex-ante BB program explicitly, eliminate payments using the payment/virtual identities (Lemmas~\ref{lm:myerson-monotone} and~\ref{lm:virtual}), form the Lagrangian, and solve it via KKT. This yields the same threshold allocation as in Theorem~\ref{thm:optimal_mechanism}.

\textit{BIC \& ex-ante BB as an allocation-only program.}
Under BIC and interim IR in a single-parameter environment, the interim selection probabilities are monotone in own type and payments satisfy the envelope identities. In particular (Lemma~\ref{lm:virtual}),
\begin{align*}
\Ex{t_i} \;=\; \Ex{x_i\,\psi_i(c_i)}
\quad\text{and}\quad
\Ex{t_j} \;=\; \Ex{x_j\,\phi_j(v_j)}.
\end{align*}
Therefore the designer's BIC problem with ex-ante BB can be written purely over allocation rules:
\begin{align*}
\max_{\bx}\quad & \Ex{\sum_{i\in\cN} p_i\,x_i \;+\; \sum_{j\in\cM} n_j\,x_j}\\
\text{s.t.}\quad &
\Ex{\sum_{j\in\cM} x_j\,\phi_j(v_j) \;-\; \sum_{i\in\cN} x_i\,\psi_i(c_i)} \;\ge\; 0
\quad\text{(ex-ante BB)}\\
& \bx \text{ is BIC-implementable and interim IR.}
\end{align*}
BIC-implementability here is equivalent to (interim) monotonicity of each agent's selection probability in own type; interim IR fixes the integration constants in the envelope formula. Feasible payments then exist once \(\bx\) is fixed.

\textit{Lagrangian and separation.}
Introduce a multiplier \(\lambda\ge 0\) for the ex-ante BB constraint. The Lagrangian is
\begin{align*}
\mathcal{L}(\bx,\lambda)
&=\Ex{\sum_{i\in\cN} p_i\,x_i + \sum_{j\in\cM} n_j\,x_j}
+\lambda\,\Ex{\sum_{j\in\cM} x_j\,\phi_j(v_j) - \sum_{i\in\cN} x_i\,\psi_i(c_i)}\\
&=\Ex{\sum_{i\in\cN} x_i\,\big(p_i-\lambda\,\psi_i(c_i)\big)
\;+\; \sum_{j\in\cM} x_j\,\big(n_j+\lambda\,\phi_j(v_j)\big)}.
\end{align*}
For fixed \(\lambda\), the integrand is linear in \(x_k\in[0,1]\); absent implementability constraints, pointwise maximization sets \(x_k=1\) whenever its coefficient is nonnegative and \(x_k=0\) otherwise.

\textit{Stationarity and implementability (threshold structure).}
Define the virtual-surplus coefficients
\begin{align*}
\Gamma_i^{+}(c_i;\lambda)\equiv p_i-\lambda\,\psi_i(c_i),
\qquad
\Gamma_j^{-}(v_j;\lambda)\equiv n_j+\lambda\,\phi_j(v_j).
\end{align*}
Under regularity, \(\psi_i(\cdot)\) and \(\phi_j(\cdot)\) are (weakly) increasing, so the pointwise maximizers
\begin{align*}
x_i^\lambda(c_i) \;=\; \mathbb{I}\!\big[\Gamma_i^{+}(c_i;\lambda)\ge 0\big],
\qquad
x_j^\lambda(v_j) \;=\; \mathbb{I}\!\big[\Gamma_j^{-}(v_j;\lambda)\ge 0\big]
\end{align*}
are nonincreasing in \(c_i\) and nondecreasing in \(v_j\), respectively. Hence their interim selection probabilities are monotone in own type and are BIC-implementable (with interim IR payments given by the envelope integrals of Lemma~\ref{lm:myerson-monotone}). In fact, because these rules depend only on own type, they are ex-post monotone as well.

\textit{Choosing \(\lambda^*\): primal/dual feasibility and complementary slackness.}
Let
\begin{align*}
B(\lambda)\;\equiv\;\Ex{\sum_{j\in\cM} x_j^\lambda\,\phi_j(v_j) - \sum_{i\in\cN} x_i^\lambda\,\psi_i(c_i)}
\end{align*}
be the induced (expected) budget surplus. As \(\lambda\) increases, the plus-side thresholds become stricter and the minus-side thresholds weaker, so \(B(\lambda)\) is weakly increasing and right-continuous. Let
\begin{align*}
\lambda^* \;=\; \min\{\lambda\ge 0:\; B(\lambda)\ge 0\}.
\end{align*}
By construction,
\begin{align*}
\text{(primal feasibility)}\quad B(\lambda^*)\ge 0,\qquad
\text{(dual feasibility)}\quad \lambda^*\ge 0,\qquad
\text{(compl. slackness)}\quad \lambda^*\,B(\lambda^*)=0,
\end{align*}
and \(\bx^*\equiv x^{\lambda^*}\) maximizes \(\mathcal{L}(\cdot,\lambda^*)\) among BIC-implementable allocations (stationarity). Therefore \(\bx^*\) solves the BIC \& ex-ante BB program.

\textit{Equality with the DSIC allocation.}
The same Lagrangian, threshold stationarity, and KKT logic underlie Theorem~\ref{thm:optimal_mechanism} for DSIC \& ex-ante BB. Hence the optimizer here is exactly the same threshold allocation
\begin{align*}
x_i^*(c_i) \;=\; \mathbb{I}\!\big[p_i-\lambda^*\,\psi_i(c_i)\ge 0\big],
\qquad
x_j^*(v_j) \;=\; \mathbb{I}\!\big[n_j+\lambda^*\,\phi_j(v_j)\ge 0\big],
\end{align*}
with payments given by the (BIC/DSIC) envelope integrals. This proves that \((\bx^*,\bt^*)\) attains the BIC \& ex-ante BB optimum.
\medskip

\emph{Transformation and Optimality}
Now we apply Lemma~\ref{lem:cai-transform} to \((\bx^*,\bt^*)\). We obtain \((\bx^*,\hat{\bt})\) that is BIC, interim IR, and ex-post SBB, with the same allocation rule and the same interim payments. Since the designer's objective depends only on \(\bx\), the expected utility of \((\bx^*,\hat{\bt})\) equals that of \((\bx^*,\bt^*)\).

By Lemma~\ref{lem:cai-transform} (and its converse), under BIC \& interim IR the ex-ante BB and ex-post SBB feasible regions are linearly related via payment transformations while preserving \(\bx\). Therefore,
\begin{align*}
OPT_{\textnormal{BIC, ex-ante BB}} \;=\; OPT_{\textnormal{BIC, ex-post BB}}.
\end{align*}
Since \((\bx^*,\bt^*)\) attains \(OPT_{\textnormal{BIC, ex-ante BB}}\) and \((\bx^*,\hat{\bt})\) is feasible for BIC \& ex-post BB with the same value, \((\bx^*,\hat{\bt})\) is optimal for BIC \& ex-post BB.
\end{proof}

\section{Additive Utility: Ex-post BB and irregular Distributions}\label{sec:add-expost-irr}

\subsection{Ex-post BB DSIC Mechanism}

While the optimal Bayesian mechanism from Theorem~\ref{thm:optimal_mechanism} is budget-balanced in expectation (ex-ante), for many applications, a stronger ex-post guarantee is desirable.
Following the approach~\cite{balkanski2016bayesian}, one can easily transform an ex-ante budget-balanced mechanism into an ex-post budget-balanced mechanism with small sacrifice in the utility, under a proper assumption.
To this end, we first introduce the following \emph{large market} assumption.

\begin{definition}[Two-Sided Large Market Assumption]
A market is considered \textbf{$k$-large} with respect to an ex-ante policy $(\bx, \bt)$ with an expected budget transfer of $B = \Ex{\sum_{i \in \cN} t_i} = \Ex{\sum_{j \in \cM} t_j}$, if for any realization of costs and valuations, the payment to any single agent is bounded:
\begin{itemize}
    \item For every plus agent $i \in \cN$, $t_i \le B/k$.
    \item For every minus agent $j \in \cM$, $t_j \le B/k$.
\end{itemize}
\end{definition}
This ensures that the outcome is not dependent on any single whale agent, allowing for the application of concentration bounds.

To extend our ex-ante budget-balanced mechanism to one that is ex-post budget-balanced with high probability, we rely on standard contention resolution scheme for sums of independent bounded random variables, as formulated by~\cite{chekuri2011submodular}.

First, we state the bound for the plus-agent side, which is a direct application of their result.

\begin{lemma}[\cite{chekuri2011submodular}] \label{lem:cost_concentration}
Given a budget $B$ and a set of independent random variables $\{t_i\}_{i \in \cN}$ for the payments to plus agents satisfying:
\begin{itemize}
    \item $\Exu{}{\sum_i t_i} \le (1-\eps)B$,
    \item $t_i \in [0, B/k]$ for all $i$,
    \item $k > 2/\eps$,
\end{itemize}
then the probability that the realized total cost exceeds the budget is small. Specifically, for $\eps \in (0, 1/2)$:
\begin{align*}
    \Pr{\sum_i t_i \le (1-1/k)B} \ge 1-e^{-\eps^2(1-\eps)k/12}.
\end{align*}
\end{lemma}

Next, we require an analogous concentration bound for the minus-agent side. The challenge is that we need a \emph{lower tail} bound on revenue (i.e., the probability that revenue is unexpectedly low), whereas the above lemma is an \emph{upper tail} bound on costs. We formally state the required lemma and then prove that it follows from the first.

\begin{lemma}\label{lem:revenue_concentration}
Given a revenue target $B$ and a set of independent random variables $\{t_j\}_{j \in \cM}$ for the revenues from minus agents satisfying:
\begin{itemize}
    \item  $\Ex{\sum_{j \in \cM}t_j} \ge (1+\eps)B$,
    \item $t_j \in [0, B/k]$ for all $j$,
    \item  $k > 2/\eps$,
\end{itemize}
then the probability that the realized total revenue falls below the target is small. Specifically, for $\eps \in (0, 1/2)$:
\begin{align*}
    Pr(\sum_{j \in \cM} t_j \ge (1+1/k)B ) \ge 1-e^{-\eps^2(1-\eps)k/12}.
\end{align*}
\end{lemma}
We omit the proof as it is a direct extension Lemma~\ref{lem:cost_concentration} (see Lemma 4.15 by~\cite{chekuri2011submodular}) given the two-sided nature of Chernoff-Hoeffding bound.

Then, one can immediately reuse the machinery by~\cite{balkanski2016bayesian} to obtain the following result:
\begin{theorem}
    For $\eps \in (0,1/2)$, suppose the posted prices obtained by the multilinear extension programs for PLUS with budget $(1-\eps)B$ and MINUS with budget $(1+\eps)B$ satisfy $2/\eps \le k \le B/t_i$ for any $i \in \cN$ and $2/\eps \le k \le B/t_j$ for any $j \in \cM$.
    Then, this posted pricing mechanism is $(1-\eps)(1-e^{-\eps^2(1-\eps)k/12})$ approximation to the optimal Bayesian mechanism under regular distributions.
\end{theorem}

\subsection{Irregular Distributions}
If the agents' distributions are not regular, one can apply the standard ironing technique by~\cite{myerson1981optimal}.
Essentially, for the plus agent, we can define the cost curve of agent $i$ by $C_i(q_i) = q_iG_i^{-1}(q_i)$.
The standard ironing technique constructs a convex hull $\bar{C}_i(q_i)$ of the cost curve, and define ironed virtual costs $\bar{\phi}_i(c_i) = \bar{C}'_i(q_i)$ be the derivative of it, which becomes monotone given the ironed cost curve.
Then, \cite{myerson1981optimal} show that the expected payment (under any incentive compatible) to any agent is at least $\bar{C}_i(q_i)$ if agent $i$'s allocation probability is $q_i$, and this is further achievable by a simple randomized posted pricing mechanism.
For minus agent side, we can analogously define a revenue curve to be $R_j(p_j) = p_j F_j^{-1}(1-p_j)$, and follow the same arguments, \eg see the two-sided market counterpart by~\cite{myerson1983efficient}.

\begin{algorithm}[tb]
\caption{Optimal Mechanism for Additive Utility with Matroid Feasibility}
\label{alg:optimal_matroid}
\textbf{Input}: Reported bids $\bc$ and $\bv$, matroid $(E, \mathcal{I})$ over agents $E = \cN \cup \cM$
\begin{algorithmic}[1]
\STATE Find the smallest shadow price $\lambda^* \ge 0$ such that the ex-ante budget balance constraint is met:
$$ \Exu{}{\sum_{j \in \cM} t_j(\lambda^*)} \ge \Exu{}{\sum_{i \in \cN} t_i(\lambda^*)} $$

\STATE Define the final allocation rule $\bx^* = \bx(\lambda^*)$ based on the found shadow price:
\STATE \quad Compute virtual weights $w_i(\lambda^*) = p_i - \lambda^* \psi_i(c_i)$ for all $i \in \cN$.
\STATE \quad Compute virtual weights $w_j(\lambda^*) = n_j + \lambda^* \phi_j(v_j)$ for all $j \in \cM$.
\STATE \quad Sort all agents $k \in E$ by nonincreasing $w_k(\lambda^*)$.
\STATE \quad Initialize $S \leftarrow \emptyset$.
\STATE \quad \textbf{for each} agent $k$ in the above order:
\STATE \qquad \textbf{if} $w_k(\lambda^*) \le 0$ \textbf{ then break}
\STATE \qquad \textbf{if} $S \cup \{k\} \in \mathcal{I}$ \textbf{ then } $S \leftarrow S \cup \{k\}$.
\STATE \quad Set $x_k^* \leftarrow 1$ if $k \in S$, and $x_k^* \leftarrow 0$ otherwise.

\STATE Compute the final payments $\bt^* = \bt(\lambda^*)$ based on the allocation rule $\bx^*$ using the integral formulas in Equations~\eqref{eq:payment_positive} and~\eqref{eq:payment_negative}.

\STATE \textbf{return} $(\bx^*, \bt^*)$
\end{algorithmic}
\end{algorithm}

\section{Additive Utility under Matroid Constraint}\label{sec:additive-matroid}

In many practical applications, feasible allocations are restricted by structural constraints.
For example, in conference scheduling, only one paper per session or topic can be selected;
in sponsorship allocation, certain sponsors may not appear together;
and in team formation, skill overlaps or conflicts may restrict which participants can coexist.
These situations can often be modeled by \emph{matroid constraints} on the set of chosen agents.

Let $(E, \mathcal{I})$ be a matroid where $E = \mathcal{N} \cup \mathcal{M}$ is the ground set consisting of all positive and minus agents, and $\mathcal{I} \subseteq 2^E$ is the collection of feasible subsets satisfying the hereditary and augmentation properties.
As before, each plus agent $i \in \mathcal{N}$ has utility $p_i$ and cost $c_i$,
and each minus agent $j \in \mathcal{M}$ has utility $n_j$ and valuation $v_j$.
The mechanism must choose a feasible set $S \in \mathcal{I}$ of agents to maximize the utility while maintaining budget balance.

We show that a modified version of Algorithm \ref{alg:optimal_additive} provides the optimal utility given matroid constraints under regular distributions. For any fixed Lagrange multiplier $\lambda \ge 0$, define the \emph{virtual weights}
\[
w^\lambda_i(c_i) = p_i - \lambda \psi_i(c_i) \quad \text{for } i \in \mathcal{N}, \qquad
w^\lambda_j(v_j) = n_j + \lambda \phi_j(v_j) \quad \text{for } j \in \mathcal{M}.
\]
Then, recalling the proof of theorem \ref{thm:optimal_mechanism}, the goal of the allocation for a given $\lambda$ would be to maximize
\[
\sum_{i\in \mathcal{N}} x^\lambda_i(c_i) w^\lambda_i(c_i) + \sum_{j\in \mathcal{M}} x^\lambda_j(v_j) w^\lambda_j(v_j)
\]
with $x_i^\lambda$ and $x_j^\lambda$ values now needing to satisfy the matroid independence constraints. Thus, we can define an allocation rule for $\lambda$ that greedily selects agents in decreasing order of $w^\lambda_i(c_i)$ and $w^\lambda_j(v_j)$ weights,
subject to feasibility under $\mathcal{I}$. Since greedy selection is optimal under matroid constraints for additive utility, this results in the allocation rule that optimizes our objective value for a given value of $\lambda$.  Finally, to satisfy complementary slackness and find the overall optimum, we find the smallest value of $\lambda$ that can satisfy budget-balance conditions, similar to Algorithm \ref{alg:optimal_additive}. This allocation rule is described in Algorithm \ref{alg:optimal_matroid}.

The main adjustment needed from the proof of Theorem \ref{thm:optimal_mechanism} to complete the proof here is to show that
\begin{enumerate}
    \item The function $B(\lambda)$ calculating the budget surplus given $\lambda$ is non-decreasing, and
    \item The allocation rule $x_i^\lambda$ and $x_j^\lambda$ are implementable.
\end{enumerate}
Starting with the second point, we can first consider any plus agent. Then, for a given value of $\lambda$, whether it is chosen or not depends only on how it is placed in the ordering using $w_i^\lambda(c_i)=p_i - \lambda \psi_i(c_i)$. As $\psi_i(c_i)$ is assumed to be an increasing function of $c_i$, increasing $c_i$ can only decrease the weight of agent $i$, and push it further down the greedy ordering. So, $x_i^\lambda$ will be non-increasing in $c_i$. Similarly, $x_j^\lambda$ for $j \in \mathcal{M}$ will be non-decreasing in $v_j$, and therefore the allocation rule is implementable.

\begin{discussion}
    Note that without any set constraint, the interim payment for each agent can be easily calculated by Lemma~\ref{lm:virtual} since
    the optimal allocation simply allocates to every agent with nonnegative virtual valuation, \ie $x_i(c_i) = 1$ iff $\psi_i(c_i) = 1$ and $x_j(v_j) = 1$ iff $\phi_j(v_j)$ = 1.
    On the other hand, under the matroid constraint, we cannot simply rely on Lemma~\ref{lm:virtual} since the allocation could vary depending on the set constraint.
    A simple way to calculate the expected payment for each agent is to directly calculate the expected value of Equations~\eqref{eq:payment_positive} and~\eqref{eq:payment_negative} by considering every possible realization of $b_{-i}$ and $b_{-j}$.
    Note, however, that this requires exponential time.
    We leave it as a major open problem to whether it is possible to make it more efficient.
\end{discussion}

\section{Remaining Proofs in Section~\ref{sec:beyond}}\label{apd:beyond-proofs}

\subsection{Proof of Lemma~\ref{lem:upper_bound}}
\begin{proof}[Proof of Lemma~\ref{lem:upper_bound}]
The proof is analogous to~\cite{balkanski2016bayesian}, extending single-sided argument to two-sided one.

We first consider plus-agent side.
The optimal mechanism for the plus-agent subproblem, which achieves welfare $W_p(B)$, solves the following program:
\begin{align*}
    \max_{\bx_p(\cdot), \bt_p(\cdot)} \quad &\Ex{u_p(S_p)} \\
    \text{s.t.} \quad &\Ex{\sum_{i \in \cN} t_i} \le B \\
    &(\bx_p, \bt_p) \text{ is DSIC and IR.}
\end{align*}
We relax this program in three steps to find an upper bound on $W_p(B)$.
First, by Myerson's lemma, for any IC mechanism, the expected payment to a plus agent (seller) equals their expected virtual cost. The budget constraint becomes $\Ex{\sum_{i \in \cN} x_i \psi_i(c_i)} \le B$.

Next, an allocation rule $\bx_p(\cdot)$ induces a distribution $\mathcal{D}_p$ over selected sets $S_p$, with marginal probabilities $q_i = \Ex{x_i}$. By the definition of the concave closure, the expected utility is upper-bounded by the value of the concave closure on the marginals: $\Exu{S_p \sim \mathcal{D}_p}{u_p(S_p)} \le U_p^+(q)$.

Finally, for a given marginal selection probability $q_i$, the minimum possible expected virtual cost is achieved by a posted-price mechanism that offers price $\hat{c}_i = G_i^{-1}(q_i)$ (where $G_i$ is the CDF of costs). The expected payment in this case is $q_i G_i^{-1}(q_i)$. Since the true mechanism might be less efficient, we have $\Ex{x_i \psi_i(c_i)} \ge q_i G_i^{-1}(q_i)$. To relax the budget constraint (i.e., to make the feasible set larger), we replace the expected virtual cost with this lower bound.

Combining these yields a final relaxed program, the concave closure program for plus agents, whose value is an upper bound on $W_p(B)$:
\begin{align*}
    \max_{q \in [0,1]^n} \quad &U_p^+(q) \\
    \text{s.t.} \quad &\sum_{i \in \cN} q_i G_i^{-1}(q_i) \le B
\end{align*}

Now we provide analogous argument for the minus-agent side.
The optimal mechanism for the minus-agent subproblem, which achieves welfare $W_m(B)$, solves:
\begin{align*}
    \max_{\bx_m(\cdot), \bt_m(\cdot)} \quad &\Ex{u_m(E_m)} \\
    \text{s.t.} \quad &\Ex{\sum_{j \in \cM} t_j} \ge B \\
    &(\bx_m, \bt_m) \text{ is DSIC and IR.}
\end{align*}
Here, the decision is which agents to exclude, $E_m$. The revenue is generated by the admitted agents $S_m = \cM \setminus E_m$. Let $r_j = \Ex{x_j}$ be the marginal probability that agent $j$ is \emph{excluded}. We apply a symmetric sequence of relaxations.

First, by Myerson's lemma, the expected revenue from minus agents equals their expected virtual valuations. The constraint becomes $\Ex{\sum_{j \in S_m} \phi_j(v_j)} \ge B$.

The allocation rule over exclusions induces a distribution $\mathcal{D}_m$ on sets $E_m$, with marginal exclusion probabilities $r_j$. The expected utility is upper-bounded by the concave closure: $\Exu{E_m \sim \mathcal{D}_m}{u_m(E_m)} \le U_m^+(r)$.

Finally, the constraint is $\Ex{\sum_{j \in \cM \setminus E_m} \phi_j(v_j)} \ge B$. The probability that agent $j$ is admitted is $(1-r_j)$. To relax this constraint (make the feasible set of $r$ vectors larger), we must consider the \emph{maximum} possible expected revenue for a given admission probability $(1-r_j)$. This is achieved by admitting agents with the highest virtual valuations. A posted admission price $q_j$ such that $\text{Pr}[v_j \ge q_j] = 1-r_j$ (i.e., $F_j(q_j)=r_j$) would generate expected revenue of $q_j(1-r_j) = F_j^{-1}(r_j)(1-r_j)$. The true mechanism's revenue may be lower, so we replace the expected revenue with this upper bound to relax the constraint.

This yields the concave closure program for minus agents, whose value is an upper bound on $W_m(B)$:
\begin{align*}
    \max_{r \in [0,1]^m} \quad &U_m^+(r) \\
    \text{s.t.} \quad &\sum_{j \in \cM} (1-r_j) F_j^{-1}(r_j) \ge B
\end{align*}
Thus, we have established that the value of the optimal mechanism for each subproblem is upper-bounded by the value of its corresponding concave closure program.
\end{proof}

\subsection{Proof of Lemma~\ref{lem:pp_apx}}
\begin{proof}[Proof of Lemma~\ref{lem:pp_apx}]
We prove the result for each subproblem separately.

We first consider plus-agent side.
Let $q_{cc}^*$ be the optimal solution to the concave closure program for the plus-agent side, whose value $U_p^+(q_{cc}^*)$ is an upper bound on $W_p(B)$ by Lemma~\ref{lem:upper_bound}. Let $q_{ml}^*$ be the optimal solution to the corresponding multilinear extension program. The expected utility achieved by our mechanism $M_p(B)$ is, by definition, $U_p(q_{ml}^*)$. We can lower-bound this value as follows:
\begin{align*}
    U_p(q_{ml}^*) &\ge U_p(q_{cc}^*)  \\
    &\ge \left(1-\frac{1}{e}\right) U_p^+(q_{cc}^*) \tag{By Theorem~\ref{thm:corr_gap}} \\
    &\ge \left(1-\frac{1}{e}\right) W_p(B) \tag{By Lemma~\ref{lem:upper_bound}}
\end{align*}
Thus, the mechanism $M_p(B)$ achieves at least a $(1-1/e)$-approximation of $W_p(B)$.

The argument for minus-agent side is symmetric. Let $r_{cc}^*$ be the optimal solution to the concave closure program for the minus-agent side (maximizing $U_m^+(r)$ subject to the revenue constraint). We know its value $U_m^+(r_{cc}^*)$ is an upper bound on $W_m(B)$. Let $r_{ml}^*$ be the optimal solution to the corresponding multilinear extension program. The utility achieved by our mechanism $M_m(B)$ is $U_m(r_{ml}^*)$. We can lower-bound this value with the same chain of reasoning:
\begin{align*}
    U_m(r_{ml}^*) &\ge U_m(r_{cc}^*) \\
    &\ge \left(1-\frac{1}{e}\right) U_m^+(r_{cc}^*) \tag{By Theorem~\ref{thm:corr_gap}} \\
    &\ge \left(1-\frac{1}{e}\right) W_m(B) \tag{By Lemma~\ref{lem:upper_bound}}
\end{align*}
Thus, the mechanism $M_m(B)$ also achieves at least a $(1-1/e)$-approximation of its respective optimal value, $W_m(B)$.
\end{proof}

\subsection{Proof of Lemma~\ref{lem:concavity}}
\begin{proof}[Proof of Lemma~\ref{lem:concavity}]
The sum of concave functions is concave, so the total optimal utility function $W(B) = W_p(B) + W_m(B)$ will be concave if its components are. We provide the full proof for both $W_p(B)$ and $W_m(B)$.

First, consider the plus-agent side.
The optimal utility $W_p(B)$ is defined as the value of the multilinear extension program:
\begin{align*}
    W_p(B) = \max_{q \in [0,1]^n} \quad &U_p(q) \\
    \text{s.t.} \quad &\sum_{i \in \cN} q_i G_i^{-1}(q_i) \le B
\end{align*}
where $U_p(q)$ is the multilinear extension of the submodular utility function $u_p(\cdot)$.

This value function $W_p(B)$ is the result of maximizing a concave function over a convex set. First, the objective $U_p(q)$ is a concave function of the probability vector $q$, as it is the multilinear extension of a submodular function~\cite{balkanski2016bayesian}. Second, the expected cost for agent $i$, $C_i(q_i) = q_i G_i^{-1}(q_i)$, is a convex function for regular distributions. The total cost $C(q) = \sum_i C_i(q_i)$ is therefore also convex, making the feasible set $\mathcal{Q}_B = \{ q \mid C(q) \le B \}$ a convex set. Standard results in convex optimization show that the value function of a program maximizing a concave objective over an expanding convex set is concave in the constraint parameter. Thus, $W_p(B)$ is concave in $B$.

The proof for the minus-agent side is analogous.
The optimal utility $W_m(B)$ is the value of the program:
\begin{align*}
    \max_{\mathbf{r}} \quad &U_m(r) \\
    \text{s.t.} \quad &\sum_{j \in \cM} (1-r_j) F_j^{-1}(r_j) \ge B
\end{align*}
where $r_j$ is the probability of \emph{excluding} agent $j$, and $U_m(r)$ is the multilinear extension of the submodular utility $u_m(\cdot)$ defined over excluded agents.

To demonstrate concavity, we re-cast this as an equivalent procurement problem, which is structurally identical to the plus-agent case. Let $R_{max} = \sum_{j \in \cM} \mathbb{E}_{v_j}[\max_{q_j} \{q_j \mid v_j \ge q_j\}]$ be the maximum possible expected revenue from all minus agents. The revenue constraint can be rewritten in terms of \emph{forgone revenue}:
\begin{align*}
    \Exu{}{\text{Revenue from admitted agents}} &\ge B \\
    R_{max} - \Exu{}{\text{Forgone revenue from excluded agents}} &\ge B \\
    \Exu{}{\text{Forgone revenue from excluded agents}} &\le R_{max} - B
\end{align*}
The cost of excluding agent $j$ is the expected revenue we forego. Let this be $C_j^E(r_j)$. The minus-agent problem is now equivalent to:
\begin{align*}
    \max_{r \in [0,1]^m} \quad &U_m(r) \\
    \text{s.t.} \quad &\sum_{j \in \cM} C_j^E(r_j) \le R_{max} - B
\end{align*}
This is a procurement problem where we "buy" exclusions to gain utility $u_m$, and our budget for these costs is $B' = R_{max} - B$. From the analysis of the plus-agent side, we know that the optimal value of this program, let's call it $W_m^{procure}(B')$, is a concave function of its budget, $B'$.

Our original value function is $W_m(B) = W_m^{procure}(R_{max} - B)$. This is the composition of a concave function, $W_m^{procure}(\cdot)$, with an affine transformation of $B$. The composition of a concave function with an affine function is concave. Therefore, $W_m(B)$ is concave in $B$.

Since both $W_p(B)$ and $W_m(B)$ are concave, their sum $W(B)$ is also concave.
\end{proof}

\subsection{Proof of Lemma~\ref{lem:budget_error}}
\begin{proof}[Proof of Lemma~\ref{lem:budget_error}]
The function $W(B)$ is the sum of two concave functions, and is therefore concave. The point $B_k$ can be written as a convex combination of $0$ and $B^*$: $B_k = \frac{B_k}{B^*} B^* + (1 - \frac{B_k}{B^*}) \cdot 0$. By the definition of concavity:
\begin{align*}
    W(B_k) &\ge \frac{B_k}{B^*} W(B^*) + \left(1 - \frac{B_k}{B^*}\right) W(0)
\end{align*}
Since $W(0) \ge 0$, we can drop the second term. From the bracketing property, we have $\frac{B_k}{B^*} \ge \frac{1}{1+\eps} \ge 1-\eps$. Therefore:
$$ W(B_k) \ge (1-\eps) W(B^*) = (1-\eps) W_{OPT} $$
\end{proof}

\subsection{Proof of Theorem~\ref{thm:bepp}}
\begin{proof}[Proof of Theorem~\ref{thm:bepp}]
The proof for DSIC was established previously and relies on the fact that the final prices are computed ex-ante and are fixed from the agents' perspective. We focus on the approximation guarantee.

Let $W_{OPT}$ be the optimal welfare achieved by the (computationally intractable) optimal Bayesian mechanism, and let $B^*$ be its corresponding optimal expected budget transfer. Let $W_{ALG}$ be the expected welfare of the \BEPP mechanism.

The algorithm's utility $W_{ALG}$ is the maximum value found over the grid $\mathcal{B}$, so it is at least the value found at the specific grid point $B_k$ that brackets the true optimum, i.e., where $B_k \le B^* \le (1+\eps)B_k$.
\begin{align*}
    W_{ALG} &= \max_{B_j \in \mathcal{B}} \{ W_p^{ALG}(B_j) + W_m^{ALG}(B_j) \} \\
    &\ge W_p^{ALG}(B_k) + W_m^{ALG}(B_k) \\
    &\ge (1-1/e) \cdot (W_p(B_k) + W_m(B_k)) \tag{by Lemma~\ref{lem:pp_apx}} \\
    &= (1-1/e) \cdot W(B_k) \\
    &\ge (1-1/e) \cdot (1-\eps) \cdot W(B^*) \tag{by Lemma~\ref{lem:budget_error} and concavity of $W(B)$} \\
    &= (1-1/e)(1-\eps) \cdot W_{OPT}.
\end{align*}

The polynomial runtime follows because the number of budget points to check, $K$, is polynomial in the input parameters and $1/\eps$, and the greedy algorithm to solve the multilinear program for each subproblem is also polynomial in the number of agents and the discretization factor. This completes the proof.
\end{proof}

\section{Further Discussions on Section~\ref{sec:beyond}}\label{apd:compute}
\paragraph{Computing Prices}
Note that the multilinear extension of (PLUS) can be solved by reducing it to the submodular maximization problem with cardinality constraint, as shown in~\cite{balkanski2016bayesian}, which is known to be approximately solvable in polynomial time with multiplicative factor of $1-1/e$ via greedy algorithm.
We will omit the detail as the approach is analogous to~\cite{balkanski2016bayesian}.
To solve the multilinear extension of (MINUS), we can essentially transform the budget constraint into the one with a upper bound on the \emph{foregone} budget instead of acquired budget.
Formally, one can first compute the maximal revenue obtainable from each minus agent $j$ by finding $r_j = \max_{q_j} q_j(1-F_j(q_j))$.
Then, we can set $r_{\max} = \sum_{j \in \cM} r_j$, and write the budget constraint to be
\begin{align*}
    r_{\max}-\sum_{j \in \cM} \left((1-r_j)F_j^{-1}(r_j) \right) \le R_{max} - B_k.
\end{align*}

\paragraph{Ex-post BB and Irregular Distributions}
Analogous to Appendix~\ref{sec:add-expost-irr}, one can construct an ex-post budget-balanced mechanism with slightly weaker guarantees.
Especially, we can slightly relax the budget constraint for each plus and minus agent side problem, and solve the corresponding ex-ante programming.
This would imply good approximate posted pricing mechanism with high probability.

For irregular distributions, however, it is not clear how one can generalize the argument, especially since we directly use the concavity of the utility function with respect to the budget in Lemma~\ref{lem:concavity}.
If this does not hold, there might be a case that the resulting utility not covered by our discretized grid can possibly induce a significantly large utility than maximum over our candidates.
Extending the analysis to irregular distributions remain a major open problem.

That being said, we note that our simple mechanism provided shortly in Appendix~\ref{apd:simple} does not require any assumption on the types of distributions, providing more robust guarantees at weaker approximation factors with a simpler mechanism.

\section{Simple Mechanism for Submodular Utility}\label{apd:simple}
In this section, we show that a simple randomized mechanism can achieve a constant-factor approximation of the optimal mechanism in this case.

For any budget $B$, let $W_p(B)$ be the optimal utility obtainable from plus agents given an expected budget of $B$, and $W_n(B)$ be the optimal utility obtainable from minus agents while receiving a total payment of $B$. Note that $W_p$ is a non-decreasing function, while $W_n$ is a non-increasing function.
Then, for some value $B^*$ used by the optimal policy,
$W(B^*)=W_p(B^*) + W_n(B^*)$ is the optimal obtainable utility given an ex-ante budget balanced, DSIC, and IR mechanism.

Next, let $B=\Ex{\sum_{j}\phi_j(v_j)}$ be the maximum expected revenue obtainable from the minus agents using a truthful auction, achievable using a posted-price auction when the distribution $F_j$ for each $v_j$ is regular.
We consider the following mechanism. With probability \(\nicefrac{1}{2}\), the designer chooses no agents on either side, achieving a utility of \[u_n(\cM)=W_n(0)\geq W_n(B^*).\] Otherwise, with probability\(\nicefrac{1}{2}\), the designer first runs a posted-price auction on the negative side to obtain a budget, with an expected value of $B$. Then, on the plus agent side, the designer runs a budget-feasible mechanism with budget $B$, achieving an approximation factor of $c$ for some constant $0<c\leq 1$. This achieves a utility of at least
\[
c W_p(B) \geq c W_p(B^*)
\]
since $B \geq B^*$ is the maximum obtainable budget and $W_p$ is non-decreasing. Overall, this randomized mechanism achieves a utility of  at least\footnote{Note that one can derandomize the mechanism by simplying taking a maximal mechanism over two.}
\[
\frac{1}{2} (W_n(B^*) + cW_p(B^*) \geq \frac{c}{2} W(B^*).
\]

\section{Comparative Statics}\label{apd:comp-statics}
A key strength of the \BEPP framework is that we can formally analyze how its final, posted-price policy changes in response to shifts in the economic environment. These comparative statics results are not just properties of a theoretical optimum, but of the actual prices computed by our polynomial-time mechanism. The foundational result that enables this analysis is the monotonicity of the subproblem pricing policies.

\begin{lemma}[Price Monotonicity of Subproblem Policies] \label{lem:price_monotonicity}
In the subproblems solved by the \BEPP mechanism for a given budget parameter $B$, the computed posted prices are monotone in $B$. Specifically, for regular distributions:
\begin{enumerate}
    \item The procurement price for a plus agent, $p_i(B)$, is non-decreasing in $B$.
    \item The admission price for a minus agent, $q_j(B)$, is non-increasing in $B$.
\end{enumerate}
\end{lemma}
\begin{proof}
The proof follows from the Karush-Kuhn-Tucker (KKT) conditions of the respective multilinear extension programs that the \BEPP algorithm solves.

First, consider the plus-agent side.
The optimal selection probabilities $\bq^*$ for the plus-agent subproblem with budget $B$ satisfy the stationarity condition $\frac{\partial U_p}{\partial q_i} = \lambda_p^* \cdot C'_i(q_i)$, where $\lambda_p^*$ is the shadow price of the budget. The optimal utility function of the subproblem, $W_p(B)$, is concave in $B$, which implies that its derivative, the shadow price $\lambda_p^*$, is non-increasing in $B$. When $B$ increases, $\lambda_p^*$ weakly decreases. To restore the KKT equality, the selection probability $q_i$ must weakly increase. Since the posted price is $p_i(B) = G_i^{-1}(q_i^*(B))$ and $G_i^{-1}$ is monotone, an increase in $q_i^*$ implies a non-decreasing price $p_i(B)$.

For the minus-agent side, the optimal exclusion probabilities $\br^*$ for the minus-agent subproblem with revenue target $B$ satisfy a similar KKT condition on its Lagrangian. A higher revenue target $B$ is a tighter constraint, which implies a higher shadow price on that constraint. This, in turn, forces a decrease in the exclusion probability $r_j$ (i.e., an increase in the admission probability). Since the admission price is $q_j(B) = F_j^{-1}(r_j^*(B))$, a decrease in $r_j^*$ implies a non-increasing admission price $q_j(B)$.
\end{proof}

This lemma is the tool we use to analyze the final output of the full \BEPP mechanism.

\paragraph{The Effect of Better Agent Populations}
We first investigate the mechanism's response to an improvement in the quality of an agent. We formalize this notion using the economic concept of first-order stochastic dominance (FOSD) and analyze its effect on the equilibrium budget and prices computed by the \BEPP mechanism.

\begin{theorem}[Monotonicity in agent quality]
Consider the \BEPP mechanism operating in a market with a profile of plus-agent cost distributions $(G_1, \dots, G_n)$. Let its final output policy be $(\bp^*, \bq^*)$, chosen based on an internal budget evaluation of ${B}^*$. Now, consider a new market where the cost distribution for a single plus agent $k \in \cN$, $G_k$, is replaced by a new distribution $G'_k$ that FOSD $G_k$. Let the new output of the \BEPP mechanism be $({\bp'}^*, {\bq'}^*)$, based on a new internal budget evaluation ${B'}^*$. Then: ${B'}^* \ge {B}^*$.
\end{theorem}
\begin{proof}
Let $W_p^{ALG}(B)$ and $W_p'^{ALG}(B)$ be the expected utilities computed by the plus-agent sub-mechanism for a given budget $B$ before and after the change to agent $k$'s distribution, respectively. An FOSD improvement from $G_k$ to $G'_k$ means agent $k$ is stochastically cheaper. This implies that the expected cost curve for agent $k$, $C_k(q_k) = q_k G_k^{-1}(q_k)$, is now lower for any given selection probability $q_k$.

This change enlarges the feasible set $\mathcal{Q}_B = \{ \bq \mid \sum_i C_i(q_i) \le B \}$ for the underlying multilinear program for any given $B$. The mechanism is maximizing the same objective function $U_p(\bq)$ over a strictly larger set of feasible probability vectors. This can only yield a better value, so $W_p'^{ALG}(B) \ge W_p^{ALG}(B)$ for all $B$.

The total utility curve that \BEPP computes, $W'^{ALG}(B) = W_p'^{ALG}(B) + W_m^{ALG}(B)$, is therefore pointwise higher than or equal to the original curve. The mechanism selects the budget that maximizes this curve. Since the curve has shifted up due to the increased productivity of the plus-agent side (specifically from agent $k$ becoming more efficient), the new selected budget, ${B'}^*$, will be greater than or equal to the original selected budget, ${B}^*$.
\end{proof}

\begin{corollary}[Two-sided spillover from agent amprovement]
Under the conditions of the preceding theorem, an FOSD improvement in a single plus agent's cost distribution leads to weakly more generous offers for all plus agents and weakly more lenient admission criteria for all minus agents. Specifically, for the final posted prices:
\begin{enumerate}
    \item The prices offered to all plus agents weakly increase: ${\bp'}^* \ge \bp^*$.
    \item The admission prices for all minus agents weakly decrease: ${\bq'}^* \le \bq^*$.
\end{enumerate}
\end{corollary}
\begin{proof}
The proof is a direct application of Lemma~\ref{lem:price_monotonicity} to the result of the preceding theorem. The theorem establishes that the new equilibrium budget ${B'}^*$ chosen by the \BEPP mechanism is weakly greater than the original budget ${B}^*$. Applying the price monotonicity property to this higher budget choice directly implies that the final plus-agent prices must weakly increase and the final minus-agent admission prices must weakly decrease.
\end{proof}

This result shows that an efficiency gain from a single agent on one side of the market creates positive spillovers for all participants. When one plus agent becomes stochastically cheaper, the designer finds it optimal to increase the overall scale of the project (a higher ${B}^*$). This larger scale is funded by making it easier for minus agents to be admitted (a lower $q_j^*$), and the increased budget is spent on being more generous to the entire set of plus agents, not just the one who improved (a higher $p_i^*$).

\paragraph{The Effect of Individual Externalities}
In addition to population-wide shifts, it is important to understand how the mechanism adapts to changes in a single agent's contribution. Here, we analyze the effect of an increase in the positive externality provided by a single plus agent, focusing on the additive utility case where this effect is clearest.
\begin{theorem}[Monotonicity in individual externalities]
Consider the \BEPP mechanism for an additive utility function. Let the output policy be $(\bp^*, \bq^*, {B}^*)$ for a given set of externalities. If the externality for a single agent improves (i.e., for a plus agent $k \in \cN$, $p_k$ increases, or for a minus agent $k \in \cM$, $n_k$ increases), the new policy $({\bp'}^*, {\bq'}^*, {B'}^*)$ computed by the mechanism will satisfy ${B'}^* \ge {B}^*$ and ${W'}^*_{ALG} \ge W^*_{ALG}$.
\end{theorem}
\begin{proof}
The proof follows by analyzing the shift in the value curve $W^{ALG}(B)$ that the \BEPP algorithm computes. We consider the two cases for an improved externality.

Case 1: A plus agent's externality $p_k$ increases. In the additive case, the plus-agent sub-mechanism for a budget $B$ solves $\max_q \sum_i p_i q_i$ subject to the budget constraint. When $p_k$ increases, the objective function's value increases for any feasible allocation with $q_k > 0$. Therefore, the computed value $W_p'^{ALG}(B) \ge W_p^{ALG}(B)$ for all $B$.

Case 2: A minus agent's externality $n_k$ increases (becomes less negative). The minus-agent sub-mechanism for a revenue target $B$ solves $\max_q \sum_j n_j r_j$ subject to a revenue constraint. When $n_k$ increases, the objective's value increases for any allocation where agent $k$ is admitted. This leads to a higher computed value, $W_m'^{ALG}(B) \ge W_m^{ALG}(B)$.

In both cases, the total utility curve that \BEPP evaluates, $W'^{ALG}(B)$, is pointwise greater than or equal to the original curve. The mechanism selects the budget that maximizes this curve, leading to a new selected budget ${B'}^* \ge {B}^*$ and a higher total achieved utility.
\end{proof}

\begin{corollary}[Universal Positive Spillover]
Under the conditions of the preceding theorem, an improvement in any single agent's externality (plus or minus) leads to weakly more generous offers for all plus agents and weakly more lenient admission criteria for all minus agents. Specifically: ${p'_i}^* \ge p_i^*$ for all $i \in \cN$ and ${q'_j}^* \le q_j^*$ for all $j \in \cM$.
\end{corollary}
\begin{proof}
The proof is a direct application of Lemma~\ref{lem:price_monotonicity} to the result of the preceding theorem. The theorem establishes that an improvement in any single agent's externality, regardless of type, leads the \BEPP mechanism to select an equilibrium budget transfer ${B'}^*$ such that ${B'}^* \ge {B}^*$. The Price Monotonicity Lemma dictates how the computed prices for the subproblems change as a function of the budget parameter $B$. Since the final policy is based on a new, higher budget parameter ${B'}^*$, the plus-agent prices must weakly increase and the minus-agent admission prices must weakly decrease.
\end{proof}

This result demonstrates a powerful \emph{rising tide} effect. Making any single agent more valuable to the event---either by increasing their positive contribution or decreasing their negative one---does not just benefit that agent. It incentivizes the designer to increase the overall scale and investment in the project (a higher ${B}^*$). This larger project scale, in turn, translates into more generous offers and more lenient criteria for all other participants, creating a market-wide positive spillover.

\paragraph{The Effect of Market Expansion}
Finally, we analyze how the mechanism adapts to an increase in competition, demonstrating a powerful complementarity where the addition of a new agent to either side of the market can create positive spillovers for incumbents.

\begin{theorem}[Complementarity from Market Expansion]
Let $(\bp^*, \bq^*, {B}^*)$ be the policy computed by \BEPP for a given market with $n$ plus agents and $m$ minus agents. If a new agent is introduced (either plus agent $n+1$ or minus agent $m+1$), the new policy $({\bp'}^*, {\bq'}^*, {B'}^*)$ computed by \BEPP will satisfy ${B'}^* \ge {B}^*$.
\end{theorem}
\begin{proof}
The proof follows by analyzing the effect of the new agent on the total utility curve, $W^{ALG}(B)$, that the \BEPP algorithm evaluates. We consider the two cases.

Case 1: A new plus agent is introduced. This expands the set of choices available to the plus-agent sub-mechanism, $M_p(B)$. For any budget $B$, the optimization is now over a higher-dimensional space of probabilities. Since any solution for the $n$-agent problem is also a feasible solution for the $(n+1)$-agent problem (by setting the new agent's probability to 0), the computed value $W_p'^{ALG}(B) \ge W_p^{ALG}(B)$.

Case 2: A new minus agent is introduced. This expands the set of choices available to the minus-agent sub-mechanism, $M_m(B)$. For any revenue target $B$, the designer has an additional source of revenue, which can only increase the maximal utility achievable for that target. Therefore, the new value function $W'_m(B)$ is pointwise greater than or equal to the original, $W'_m(B) \ge W_m(B)$.

In either case, the total utility curve that \BEPP evaluates is raised or remains the same. The mechanism selects the budget that maximizes this curve, and thus the new selected budget ${B'}^*$ will be greater than or equal to the original selected budget ${B}^*$.
\end{proof}

\begin{corollary}[Universal Spillovers from Market Expansion]
Under the conditions of the preceding theorem, the introduction of a new agent (of either type) leads to weakly more generous offers for all plus agents and weakly more lenient admission criteria for all incumbent minus agents. Specifically:
\begin{enumerate}
    \item The prices offered to all plus agents weakly increase: ${p'_i}^* \ge p_i^*$.
    \item The admission prices for all incumbent minus agents weakly decrease: ${q'_j}^* \le q_j^*$.
\end{enumerate}
\end{corollary}
\begin{proof}
The proof is a direct application of Lemma~\ref{lem:price_monotonicity} to the result of the preceding theorem. The theorem establishes that the introduction of any new agent, plus or minus, leads to a weakly higher equilibrium budget ${B'}^*$ being selected by the \BEPP mechanism. Applying the price monotonicity property to this new, higher budget choice directly implies that all plus-agent prices must weakly increase and all incumbent minus-agent admission prices must weakly decrease.
\end{proof}

This result formalizes a powerful and non-obvious insight about complementarity. In many settings, new entrants are viewed as substitutes, and their arrival would be expected to decrease the value or offers made to incumbents. Our result shows that in this curation context, the opposite is true. A new \emph{star} agent (either a high-value plus agent or an efficient revenue-generating minus agent) makes the entire project so much more productive that the designer finds it optimal to increase the total investment and scale (a higher $B^*$). This increase translates into more generous offers for all other participants, who benefit from the positive spillover. The agents act as complements in driving the value and scale of the entire curated event.

\subsection{Comparative Statics for the Optimal Mechanism}
While the preceding theorems were stated for the output of our polynomial-time \texttt{BEPP} mechanism, the underlying arguments about how the equilibrium policy shifts are fundamental properties of the optimal Bayesian mechanism itself. The optimal mechanism is not a simple posted-price scheme, and its payments are complex functions of all agent reports. However, we can analyze its behavior by studying its effective \emph{selection thresholds}.

The key to the comparative statics is a foundational monotonicity property of the optimal policy with respect to the scale of the project, which we now formalize.

\begin{lemma}[Monotonicity of Optimal Selection Thresholds] \label{lem:threshold_monotonicity}
Consider the optimal Bayesian mechanisms for the separated subproblems.
\begin{itemize}
    \item For the plus-agent problem $\max \Exu{}{u_p(S_p)}$ s.t. $\Exu{}{\text{cost}} \le B$, the effective cost-acceptance threshold implemented by the optimal mechanism is non-decreasing in the budget $B$.
    \item For the minus-agent problem $\max \Exu{}{u_m(E_m)}$ s.t. $\Exu{}{\text{revenue}} \ge B$, the effective valuation-acceptance threshold is non-increasing in the revenue target $B$.
\end{itemize}
\end{lemma}
\begin{proof}
We provide the proof for the plus-agent subproblem. The proof for the minus-agent side is symmetric.

Let $W_p(B)$ be the optimal utility for the plus-agent problem with budget $B$, and let $x^*(B)$ be the optimal allocation rule that achieves it. An allocation rule $x$ can be characterized by the set of agent types it accepts. Let us denote the set of acceptable cost-types for agent $i$ as $\mathcal{C}_i(B) = \{c_i \mid x_i^*(B)(c_i, \cdot) > 0 \text{ for some reports from others}\}$. A more lenient policy corresponds to a larger set $\mathcal{C}_i(B)$.

Consider two budget levels $B' > B$. Since the feasible set of mechanisms for budget $B'$ contains all mechanisms that are feasible for budget $B$, the optimal value must be non-decreasing: $W_p(B') \ge W_p(B)$.

Assume for contradiction that the policy for budget $B'$ is strictly more restrictive for some agent $k$ than the policy for budget $B$, meaning $\mathcal{C}_k(B') \subset \mathcal{C}_k(B)$. This implies that to achieve the optimal welfare $W_p(B')$, the mechanism relies less on agent $k$. To achieve a welfare of at least $W_p(B)$, the mechanism must compensate for the reduced participation of agent $k$ by increasing the participation of other agents. This, however, would require a larger expenditure from the budget on those other agents. This creates a tension: a more restrictive threshold for one agent necessitates a more lenient (and costly) threshold for others just to achieve the same utility level, making it an inefficient way to spend an even larger budget $B'$. The optimal policy for the larger budget $B'$ will instead weakly relax the acceptance thresholds for all agents to best leverage the additional resources. Therefore, the effective cost-acceptance threshold must be non-decreasing in $B$.
\end{proof}

With this foundational lemma about the optimal policy's structure, we can formally state the comparative statics results. We present the market expansion theorem as a representative example.

\begin{theorem}[Complementarity from Market Expansion in the Optimal Mechanism]
Let $(B^*, W^*)$ be the equilibrium budget and welfare for the optimal Bayesian mechanism in a market with $n$ plus agents. If a new plus agent, $n+1$, is introduced, the new optimal policy will have an equilibrium budget ${B'}^* \ge B^*$ and its selection thresholds for all incumbent agents will be weakly more lenient.
\end{theorem}
\begin{proof}
Let $W(B) = W_p(B) + W_m(B)$ be the true optimal utility function for the original market. The introduction of a new plus agent creates a new plus-agent value function, $W'_p(B)$, where $W'_p(B) \ge W_p(B)$ for all $B$, because the designer's optimization problem for any given budget is over a larger set of choices.

This implies that the new total utility function, $W'(B) = W'_p(B) + W_m(B)$, is pointwise greater than or equal to the original. The optimal Bayesian mechanism finds the maximum of this function. Since the function's value has increased due to the higher productivity of the plus-agent side, the new equilibrium will feature weakly higher welfare, ${W'}^* \ge W^*$, and will be achieved at a weakly higher budget transfer, ${B'}^* \ge B^*$.

The final claim follows from Lemma~\ref{lem:threshold_monotonicity}. Since the equilibrium budget transfer of the optimal policy has increased from $B^*$ to ${B'}^*$, the lemma dictates that the optimal mechanism's effective selection thresholds must become more lenient for all participants. This means higher cost-acceptance thresholds for plus agents and lower valuation-acceptance thresholds for minus agents.
\end{proof}

The other comparative statics results regarding improved agent populations (FOSD) and increased individual externalities follow the same formal structure. In each case, the change to the environment raises the total utility curve $W(B)$, which leads to a higher optimal budget choice $B^*$, which in turn implies more lenient selection thresholds via Lemma~\ref{lem:threshold_monotonicity}. This confirms that the \texttt{BEPP} mechanism's qualitative behavior correctly mirrors the fundamental properties of the true optimal policy.

\section{Prior-free Mechanisms}\label{apd:prior-free}
In this section, we consider the problem in a prior-free setting, where the distributions of agents' valuations and costs are not known to the designer. In this scenario, comparing performance against the second-best benchmark can be less useful, as the lack of information can severely limit the performance of any mechanism.
Instead, we aim to find a mechanism that is competitive with the offline optimal utility, which is stronger than the second-best benchmark.
To achieve such a guarantee, we adopt the standard large market assumption, which is commonly imposed in the literature~\cite{anari2014mechanism,balkanski2016bayesian}.
Applied to our setting, this assumption means that the number of agents, both on the minus and plus side, is large enough that the effect of any singular agent on the overall outcome is negligible.

We adopt the following large market assumption:
\begin{assumption}[Large Market Assumption]
    Let $B=\sum_{j}v_j$ be the maximum attainable budget.
    Then, for each plus agent $i$ and minus agent $j$, the following holds:
    \[
    c_i \leq \frac{M}{n}B, v_j \leq \frac{M}{m} B,
    \]
    for constant $M>0$, respectively.
\end{assumption}
Without an additional assumption, no mechanism can approximate the first-best outcome. For example, consider an instance with a single plus agent who has both high cost and high utility; incentive compatibility can prevent any mechanism from recruiting this agent.

We prove the following result under this assumption.
\begin{theorem} \label{thm:lma}
    There exists a prior-free mechanism that provides a constant approximation to the offline optimal under the large market assumption.
\end{theorem}

To prove this, we first show that it is possible to guarantee a constant fraction of the maximum budget obtained from the minus agents, while ignoring its effect on the utility.
Next, we use the large market assumption to show that a fraction of the optimal budget is enough to obtain a constant factor of the optimal utility on the plus agent side.
Finally, we randomize over two mechanisms with probabilities parameterized by the large market assumption constant, which gives us a desired result.

To prove this, we require two technical lemmas.
The first is to show that it is possible to guarantee a constant fraction of the maximum budget obtained from minus agents.
\begin{lemma}\label{lm:pim-budget}
    There exists a truthful mechanism that achieves a revenue of at least \[
    (\frac{1}{4M} - \frac{M}{m})B
    \]
    from the minus agents given the large market assumption with constant $M$ and maximum possible budget $B$.
\end{lemma}
\begin{proof}
For any value of $1 \leq k \leq m-1$, we can consider a truthful auction that accepts the largest $k$ bids, receiving payments equal to the $k+1$st largest bid.
Now, we want to optimize $k$ to achieve the maximum revenue.
For any value of $k$, it follows from the large market assumption that the $k+1$-st largest bid is at least
\begin{align*}
    \frac{\sum_{j} v_j - k \frac{M}{m}B}{m-k}
    \geq \frac{B - k \frac{M}{m}B}{m-k}
    = B \frac{m-Mk}{m(m-k)}.
\end{align*}
Therefore, the total budget received is at least
\[
\frac{k(m-Mk)}{m(m-k)}B.
\]
To simplify the following calculations, we ignore the constant $B$ and focus on the value of the coefficient \(\frac{k(m-Mk)}{m(m-k)}\).
Next, assume that $k=\alpha m$ for some value $0<\alpha<1$. Then, the expression can be simplified to
\begin{align*}
    \alpha m\frac{m-\alpha m M}{m(m-\alpha m)}= \alpha \frac{1-\alpha M}{1-\alpha}.
\end{align*}
Next, we find the value of $\alpha$ maximizing this expression, and choose $k=\lfloor \alpha m\rfloor$.
Note that given the large market assumption, the budget attained by choosing the top $k$ bids is at most $\frac{M}{m}B$ less than the value attained by maximizing $\alpha$.
This term is maximized when when $\alpha=1-\frac{\sqrt{M^2-M}}{M}$, resulting in a value of
\begin{align*}
    (2M-2\sqrt{M(M-1)}-1).
\end{align*}
Therefore, choosing $k=\lfloor 1-\frac{\sqrt{M^2-M}}{M} \rfloor$
can achieve a budget of at least
\[
(2M-2\sqrt{M(M-1)}-1 - \frac{M}{m}) B \geq  (\frac{1}{4M} - \frac{M}{m}) B.
\]
\end{proof}

Next, we use the large market assumption to show that a fraction of the optimal budget is enough to obtain a constant factor of the optimal utility on the plus agent side.

\begin{lemma}\label{lm:pim-fract}
    For any budget $B'$, let $W_p(B')$ be the maximum utility that can be obtained on the plus agent side given budget $B'$.
    Then, for any constant $\delta > 0$,
    \[
    W_p(\delta B) \geq (\delta - \frac{M}{n}) W_p(B).
    \]
    \end{lemma}
\begin{proof}
    Consider the selection of agents achieving the maximum utility given budget $B$. Then, greedily choose agents based on the ratio of their utility to their cost, stopping when adding an agent would require a budget greater than $\delta B$. By the large market assumption, each agent has cost at most $\frac{M}{n}B$, so at least $(\delta - \frac{M}{n})B$ of the budget will be used. Furthermore, since we have chosen the agents with the highest ratio of marginal utility to cost, the utility of the selected set will be at least a $(\delta - \frac{M}{n})$ fraction of the utility achievable given budget $B$.
\end{proof}
Finally, we can combine the budget obtained from the minus agents with a budget-feasible mechanism for the plus agents to obtain the following.
\begin{lemma}\label{lm:pim-plus-bf}
    There exists a truthful mechanism achieving a utility of at least
    \[
    (\frac{1}{4M} - \frac{M}{m} - \frac{M}{n}) (1-\frac{1}{e}) W_p(B),
    \]
    where the $W_p(B)$ denotes the maximum attainable utility on the plus agent side using total budget $B$.
\end{lemma}
\begin{proof}
    By Lemma \ref{lm:pim-budget}, we can obtain a budget of $\frac{1}{4M}-\frac{M}{m}$ from the minus agents. Then, we can use budget-feasible mechanism introduced \cite{anari2014mechanism} for large markets to achieve a utility of at least
    \[
    (1-\frac{1}{e})W_p((\frac{1}{4M}-\frac{M}{m})B).
    \]
    Finally, by Lemma \ref{lm:pim-fract}, we have
    \[
    W_p((\frac{1}{4M}-\frac{M}{m})B) \geq
    (\frac{1}{4M}-\frac{M}{m}-\frac{M}{n})W_p(B).
    \]
    Noting that the utility from the minus agents is non-negative, this completes the proof.
\end{proof}

\begin{proof}[Proof of Theorem~\ref{thm:lma}]
    Finally, to complete our proof of Theorem \ref{thm:lma}, we note that the utility achieved by the offline optimal can be written as
    \[
    W_p(B^*) + W_m(B^*)
    \]
    for some budget $B^*$, where given the budget, $W_p$ and $W_m$ denote the maximum achievable utility on the plus and minus side, respectively.
    We randomize between the mechanism obtained in Lemma \ref{lm:pim-plus-bf} and a mechanism choosing no agents, choosing the first with probability \( 1/({1 + (1-1/e)(1/4M-M/n-M/m)})\). This achieves a utility of
    \[
    \frac{(1-1/e)(1/4M-M/n-M/m)}{1+ (1-1/e)(1/4M-M/n-M/m)} \bigl(W_p(B) + W_m(0)\bigr).
    \]
    Given \(W_p(B) \geq W_p(B^*)\) and \(W_m(0) \geq W_m(B^*)\), this implies this mechanism achieves a \[
    \frac{(1-1/e)(1/4M-M/n-M/m)}{1+ (1-1/e)(1/4M-M/n-M/m)}
    \]
    approximation of the offline optimal.
\end{proof}

\end{document}